\documentclass{article}
\usepackage[utf8]{inputenc}

\usepackage{amsmath}       
\usepackage{amsthm}        
\usepackage{amssymb}       
\usepackage{braket}       
\usepackage{amsfonts}
\usepackage{calc}
\usepackage{comment}
\usepackage{mathtools}
\usepackage{bbm}            
\usepackage{enumerate}  
\usepackage{thm-restate}
\usepackage[numbers,sort]{natbib}
\usepackage{enumitem}
\usepackage{blindtext}
\usepackage[export]{adjustbox}

\makeatletter

\theoremstyle{plain} 

\newtheorem{theorem}{Theorem}[section]
\newtheorem{lemma}[theorem]{Lemma}
\newtheorem{corollary}[theorem]{Corollary}          
\newtheorem{proposition}[theorem]{Proposition}

\newtheorem*{theorem*}{Theorem}
\newtheorem*{lemma*}{Lemma}

\theoremstyle{definition} 
\newtheorem{definition}[theorem]{Definition}
\newtheorem{remark}[theorem]{Remark}

\newtheorem{notation}[theorem]{Notation}

\theoremstyle{remark}  

\newsavebox\ideabox

\newcommand\supp{\mathrm{supp}}
\newcommand\Tr{\mathrm{Tr}}

\newcommand\Var{\mathrm{Var}}

\usepackage{cancel}
\usepackage{xcolor}
\newcommand{\be}{\begin{equation}}
\newcommand{\ee}{\end{equation}}
\newcommand{\bea}{\begin{eqnarray}}
\newcommand{\eea}{\end{eqnarray}}
\newcommand{\bes}{\begin{equation*}}
\newcommand{\ees}{\end{equation*}}
\newcommand{\beas}{\begin{eqnarray*}}
\newcommand{\eeas}{\end{eqnarray*}}
\allowdisplaybreaks
\usepackage[pdftex,colorlinks=true,linkcolor=blue,citecolor=blue,urlcolor=black]{hyperref}

\usepackage{authblk}

\usepackage{geometry}
\title{Testing quantum satisfiability}

\author[1,2]{Ashley Montanaro\thanks{ashley.montanaro@bristol.ac.uk}}
\author[3]{Changpeng Shao\thanks{changpeng.shao@amss.ac.cn}}
\author[1]{Dominic Verdon\thanks{dominic.verdon@bristol.ac.uk}}
\affil[1]{School of Mathematics, University of Bristol, Bristol BS8 1UG, United Kingdom}
\affil[2]{Phasecraft Ltd., Bristol BS1 4XE, United Kingdom}
\affil[3]{Academy of Mathematics and Systems Science, Chinese Academy of Sciences, Beijing 100190, China}

\date{\today}

\begin{document}
\maketitle

\begin{abstract}
Quantum $k$-SAT (the problem of determining whether a $k$-local Hamiltonian is frustration-free) is known to be QMA$_1$-complete for $k\ge 3$, and hence likely hard for quantum computers to solve. Building on a classical result of Alon and Shapira, we show that quantum $k$-SAT can be solved in randomised polynomial time given the `property testing' promise that the instance is either satisfiable (by any state) or far from satisfiable by a product state; by `far from satisfiable by a product state' we mean that $\epsilon n^k$ constraints must be removed before a product state solution exists, for some fixed $\epsilon>0$. The proof has two steps: we first show that for a satisfiable instance of quantum $k$-SAT, most subproblems on a constant number of qubits are satisfiable by a product state. We then show that for an instance of quantum $k$-SAT which is far from satisfiable by a product state, most subproblems are unsatisfiable by a product state. Given the promise, quantum $k$-SAT may therefore be solved by checking satisfiability by a product state on randomly chosen subsystems of constant size.
\end{abstract}

\section{Introduction}

\paragraph{Property testing.} In computer science, one often needs to determine whether a structure defined by a large quantity of data possesses a certain property which is \emph{global}, in the sense that it is defined with reference to all the data. Property testing algorithms aim to determine with high probability whether a structure possesses a global property by performing \emph{local} checks; that is, checking properties of randomly chosen small subsets of the data defining the structure. Since the time complexity of a local check is smaller than that of a global check, and one only inspects a small subset of the data, this approach, when feasible, often leads to algorithms with excellent time and query complexity. 

Of course, in general the structure may not possess the global property, but be close enough to possessing it that one is highly unlikely to perform a local check which witnesses this fact. For this reason property testing algorithms usually require a promise that the structure either possesses the property or is far from possessing the property, where distance from possessing the property is defined by a measure specific to the problem and is usually quantified by some constant $\epsilon$.

In~\cite{Alon2003}, Alon and Shapira showed that the NP-hard problem of classical $k$-SAT (satisfiability of Boolean functions on $n$ variables, where each function depends on exactly $k$ variables) is amenable to a property testing approach. Specifically, with the promise that an instance of classical $k$-SAT is either satisfiable or far from satisfiable, satisfiability may be determined in randomised polynomial time in $n$, by choosing constant-sized subsets of variables and checking satisfiability of those functions which depend only on variables in those subsets. Because of the promise, their result is limited to dense instances of $k$-SAT, where there are $\Omega(n^k)$ functions.

\paragraph{Testing quantum satisfiability.} In~\cite{Bravyi2011} Bravyi defined a quantum analogue of classical $k$-SAT, which we will here call quantum $k$-SAT. An instance of this problem is defined by $n$ Hilbert spaces of dimension $2$ (qubits) and, for every subset $s$ of $k$ qubits, a projector $\Pi_s$ on the Hilbert space of those qubits; we say that the instance is satisfiable  (``frustration-free'' in physics terminology) if there is a state of all $n$ qubits which is in the kernel of all of the projections. While classical $k$-SAT is NP-complete whenever $k \geq 3$, quantum $k$-SAT is QMA${}_1$-complete whenever $k \geq 3$, where QMA${}_1$ is the quantum analogue of the class MA with one-sided error. In the terminology of physics quantum $k$-SAT is precisely the problem of determining whether a $k$-local Hamiltonian on $n$ qubits is frustration free.

Since the problem seems very difficult in general, previous work on algorithms for quantum $k$-SAT has focused on tractable special cases. There are many interesting results on random quantum $k$-SAT, for instance~\cite{Laumann2010,Laumann2010a,Bravyi2010,Hsu2013}. It has been shown that quantum $2$-SAT can be solved in linear time~\cite{Beaudrap2016,Arad2018}; moreover, in this case there is always a satisfying state which is a product of one- and two-qubit states~\cite{Chen2011}. The case where there are relatively few nontrivial projectors all of rank 1 is also often tractable~\cite{Ambainis2012,Aldi2021}. 

In this work we also treat a special case: the dense case, where the interaction hypergraph (i.e. the $k$-uniform hypergraph on $n$ vertices whose edges are subsets of $k$ qubits whose associated projection is nontrivial) contains $\Omega(n^k)$ edges. We show that the results of~\cite{Alon2003} extend to the quantum setting. Given the promise that the instance is either satisfiable or far from satisfiable \emph{by a product state}, quantum $k$-SAT may be solved in randomised constant time in $n$ by checking satisfiability by a product state on constant-sized subsets of qubits. 

\paragraph{Some comments on our results.} 
Before stating our results formally, we will make some comments on their significance. Firstly, as we emphasised in the last paragraph, the promise is that the instance, if unsatisfiable, is far from satisfiable \emph{by a product state}. Obviously, this is a less stringent requirement than being far from satisfiable by any state. The reason we are able to weaken the promise in this way is that an instance of quantum $k$-SAT which is satisfiable by any state is locally satisfiable by a product state with high probability. This result (Theorem~\ref{thm:localsat}) is essentially a fact about quantum states which may be of independent interest; there are some connections with monogamy of entanglement~\cite{Koashi2004,Christandl2004,Brandao2011} and de Finetti theorems~\cite{Brandao2016, Brandao2013} which we discuss briefly in Remark~\ref{rem:entanglement}.

Secondly, our proofs are combinatorial in nature; we make no use of the norm. In particular, we do not require the polynomial lower bound on the ground state energy in the unsatisfiable case that is usually given as a promise in the definition of quantum $k$-SAT; our testing algorithm therefore solves a harder problem than quantum $k$-SAT as it is usually defined.

\paragraph{Connection with the quantum PCP conjecture.}
The quantum PCP (probabilistically checkable proofs) conjecture is a prominent open problem in quantum information theory~\cite{aharonov13}. The conjecture states that it is QMA-complete to approximate the ground energy of a quantum Hamiltonian up to constant relative error. That is, given a quantum Hamiltonian $H$ expressed as the sum of $m$ $k$-local terms, where $k$ and the norm of each term are $O(1)$, the conjecture states that there is some $\gamma > 0$ such that it is hard to approximate the lowest eigenvalue of $H$ up to additive error $\gamma m$.

The quantum PCP conjecture is believed to be very challenging to solve. An interesting intermediate step towards this conjecture that was recently proven by Anshu, Breuckmann and Nirkhe is the NLTS (no low-energy trivial states) theorem~\cite{anshu23}. This result states that there exist local quantum Hamiltonians on $n$ qubits, and made up of $O(n)$ terms, such that any state prepared by a quantum circuit which has depth $o(\log n)$ must have energy at least $\epsilon n$, for some constant $\epsilon > 0$. Prior to the proof of the NLTS theorem, Anshu and Breuckmann proved a ``combinatorial'' NLTS theorem~\cite{anshu22}:

\begin{theorem}[Combinatorial NLTS~\cite{anshu22}]
    There exists a constant $\epsilon > 0$ and an explicit family of $O(1)$-local satisfiable Hamiltonians $\{H_n\}$, where each Hamiltonian $H_n$ acts on $n$ qubits and is a sum of $m=\Theta(n)$ local terms $H^{(i)}_n$, such that for any family of states $\{\ket{\psi_n}\}$ satisfying
    \[ \frac{1}{m} \left|\{i: \braket{\psi_n|H^{(i)}_n|\psi_n} > 0\} \right| \le \epsilon, \]
    the depth of any quantum circuit preparing $\ket{\psi_n}$ is $\Omega(\log n)$.
\end{theorem}

Thus the combinatorial NLTS theorem is similar to standard NLTS, but is expressed in terms of the number of violated terms, rather than the total energy. One could additionally define a combinatorial version of the quantum PCP conjecture: that it is QMA$_1$-complete to distinguish, for a given local Hamiltonian $H$ made up of $m$ terms $H^{(i)}$, between the two cases a) that $H$ is satisfiable; and b) that there exists no state $\ket{\psi}$ such that $\left|\{i: \braket{\psi|H^{(i)}|\psi} > 0\} \right| \le \epsilon m$.

Our result shows that there can be no combinatorial quantum PCP theorem in the ``dense'' case where $m=\Theta(n^k)$. This is because our algorithm efficiently distinguishes between Hamiltonians that are satisfiable, and Hamiltonians where one must remove $\epsilon n^k$ terms to make the Hamiltonian satisfiable by a product state. If there exists no state $\ket{\psi}$ such that $\left|\{i: \braket{\psi|H^{(i)}|\psi} = 0\} \right| \ge (1-\epsilon)m$, then in particular there is no product state $\ket{\psi}$ such that $\left|\{i: \braket{\psi|H^{(i)}|\psi} = 0\} \right| \ge (1-\epsilon)m$, and this constraint is equivalent to it being necessary to remove $\epsilon m$ terms to make the Hamiltonian satisfiable by a product state.

It is not obvious to us whether our results could also be used to rule out dense versions of the combinatorial NLTS theorem. This is because we show that satisfiability of dense quantum $k$-SAT instances can be checked by considering only product states on a constant-sized subset of qubits, but this does not seem to automatically imply the existence of an easily preparable state satisfying the entire instance.

\subsection{Results}\label{sec:results}

We will first state the three basic definitions we will use in this work. Throughout we write $\sigma_k(S)$ for the set of $k$-element subsets of a set $S$ and $[n]$ for the set $\{1, \dots, n\}$. For any subset $x \in \sigma_k(S)$ we use the notation $\overline{x}:=S-x$ for the complement.
\begin{definition}[Quantum $k$-SAT]\label{def:qsat}
An instance of quantum $k$-SAT (without a lower bound on the ground state energy) is defined by the following data. 
\begin{itemize}
\item A number of qubits $n$. 
    \item For every $k$-element subset $s \in \sigma_k([n])$, a projector $\Pi_s$ on $(\mathbb{C}^2)^{\otimes n}$ which acts nontrivially only on the qubits in $s$. 
\end{itemize}
We write the instance as  $([n],\{\Pi_s\})$. The problem is to determine whether there exists some state $\ket{x} \in (\mathbb{C}^2)^{\otimes n}$ such that the following equation holds:
    \begin{equation}\label{eq:satisfying}
    \sum_{s \in \sigma_k([n])} \Pi_s \ket{x} = 0.
    \end{equation}
If such a state exists we say that the instance is \emph{satisfiable}; if it does not we say that the instance is \emph{unsatisfiable}. If there is a product state $\ket{\phi}$ satisfying~\eqref{eq:satisfying} we say that the instance is \emph{satisfiable by a product state}.
\end{definition}

\noindent
For precision reasons, we follow~\cite{Beaudrap2016} in requiring that the coefficients of the projectors are drawn from a finite-degree field extension $\mathbb{F}$ of the field of rational numbers. We suppose that $\mathbb{F}$ is specified as part of the problem input by a minimal polynomial $p \in \mathbb{Q}[x]$ for which $\mathbb{F} \cong \mathbb{Q}[x]/p$. We also assume a specification of how $\mathbb{F}$ embeds into $\mathbb{C}$. We furthermore require that there is some constant $K$ which bounds above the size of the specification of $\mathbb{F}$ and the size of the specification of the coefficients of the projectors. For more on this framework, see~\cite{Beaudrap2016}. In Appendix A.2 of that work,  the time complexity of arithmetic operations is analysed; all we will need here is the fact that arithmetic operations can be performed in constant time in $n$.

\noindent
The usual definition of quantum $k$-SAT includes a promise that the ground-state energy has a polynomial lower bound if the instance is unsatisfiable~\cite[\S{}1]{Bravyi2011}; Definition~\ref{def:qsat} is therefore a strictly harder problem.

\begin{definition}[$\epsilon$-far from satisfiable by a product state]\label{def:efar}
Let $\Sigma \subset \sigma_k([n])$ be such that, for some product state $\ket{\phi} \in (\mathbb{C}^2)^{\otimes n}$:
\begin{equation}\label{eq:efar}
\sum_{s \in (\sigma_k([n]) - \Sigma)} \Pi_s \ket{\phi}  =0.
\end{equation}
We say that an instance $([n],\{\Pi_s\})$ of quantum $k$-SAT is \emph{$\epsilon$-far from satisfiable by a product state} if this implies that $|\Sigma| \geq \epsilon n^k$.
\end{definition}
\begin{remark}
Our definition of $\epsilon$-far differs slightly from that in~\cite{Alon2003}; whereas they remove individual clauses (which in our case correspond to rank-one projections) we remove the whole projection for a given $k$-subset. This does not make much difference to the overall theory, since defining $\epsilon$-far in terms of the number of rank-one projections removed would only alter $\epsilon$ by a constant factor.
\end{remark}
\begin{definition}[Local satisfiability]
Let $([n],\{\Pi_s\}_{s \in \sigma_k([n])})$ be an instance of quantum $k$-SAT. Let $C \subset [n]$.

We say that the \emph{restriction} of the instance to $C$ is the instance of quantum $k$-SAT defined by the data $(C,\{\Pi_s\}_{s \in \sigma_k(C)})$. We say that $([n],\{\Pi_s\}_{s \in \sigma_k([n])})$ is \emph{locally satisfiable by a product state at $C$} if $(C, \{\Pi_s\}_{s \in \sigma_k(C)})$ is satisfiable by a product state.
\end{definition}
\noindent
Our result is based on the following two theorems, whose proofs are given in the remainder of the paper.

\begin{theorem}\label{thm:localsat}
Let $([n],\{\Pi_s\})$ be a satisfiable instance of quantum $k$-SAT, and let $c \in \mathbb{N}$ be fixed, where $c \geq 3$. Let $C \in \sigma_c([n])$ be a subset chosen uniformly at random. The probability that the instance is locally satisfiable by a product state at $C$ is greater than $p \in (0,1)$ whenever $n > \Psi(p,c)$, where 
$$\Psi(p,c) := \max\left( \frac{c^8(2^{3c+20})}{(1-\sqrt{p})^5}~,~\frac{2^{c+2}(c+1)}{-\ln(p)}\right).$$
Equivalently, let $\ket{\psi} \in (\mathbb{C}^{2})^{\otimes n}$ be any state and let $C \in \sigma_c([n])$ be a subset chosen uniformly at random. The probability that the subspace $\supp(\Tr_{\overline{C}}(\ket{\psi}\bra{\psi})) \subseteq (\mathbb{C}^2)^{\otimes c}$ contains a product state is greater than $p \in (0,1)$ whenever $n > \Psi(p,c)$.
\end{theorem}

\begin{remark}
The theorem applies to the case where $c\geq 3$; a similar result holds for $c=2$, but the result for $c=2$ is much simpler and stronger and is treated separately in Proposition~\ref{prop:c=2}. The bound $\Psi(p,c)$ is far from optimal, since we made many approximations to simplify the proof; if required, the reader could easily improve the bound by tightening the analysis without changing the structure of the proof. However, we would be surprised if the exponential dependence on $c$ could be removed without a different approach.
\end{remark}
\noindent
For the second theorem, we assume that $n$ is large enough (but still finite). 

\begin{theorem}\label{thm:localunsat}
There exists a constant $c(k,\epsilon)$ (i.e. not depending on $n$) such that the following holds for large enough $n$. Let $([n],\{\Pi_s\})$ be any instance of quantum $k$-SAT which is $\epsilon$-far from satisfiable by a product state. Then, for a randomly chosen subset $C \in \sigma_{c(k,\epsilon)}([n])$, the instance is locally unsatisfiable by a product state at $C$ with probability $p > 0.75$.
\end{theorem}
\begin{remark}
The main weakness of Theorem~\ref{thm:localunsat} is that we do not have a good upper bound for the constant $c(k,\epsilon)$. Our proof in Section~\ref{sec:localunsat} does yield an upper bound, but it is extremely large (bigger than $(((3!)!)!\cdots)!$, where the chain of factorials has length $(5/\epsilon)4^{k-1}$). The bound has this form because the degree of the $m$-qubit Segre variety is the factorial $m!$ (Lemma~\ref{lem:prodcount}). It is an open question whether the upper bound can be made small enough for the testing result to be practically useful rather than merely theoretically interesting. Improvements on this scale are not unknown in classical property testing (see e.g. the discussion in~\cite[\S{} 1.1]{Alon2003}). 
\end{remark}
\begin{remark}
Together, Theorems~\ref{thm:localsat} and~\ref{thm:localunsat} imply that, for large enough $n$, an instance $([n],\{\Pi_s\})$ of quantum $k$-SAT cannot be both satisfiable and far from satisfiable by a product state.
\end{remark}
\noindent
The consequent result about testability of quantum $k$-SAT is as follows.
\begin{corollary}\label{cor:testing}
With a promise that the instance is either satisfiable or $\epsilon$-far from satisfiable by a product state, quantum $k$-SAT may be solved in randomised polynomial time in $n$.
\end{corollary}

\begin{proof}
Let $([n],\{\Pi_s\})$ be an instance of quantum $k$-SAT, and let $c$ be some constant. Given a subset $C \in \sigma_{c}([n])$, we can check whether or not $([n],\{\Pi_s\})$ is locally satisfiable by a product state at $C$ in time polynomial in $n$ (and exponential in $c$) using Gr\"obner basis methods. Indeed, let $\{ \Pi_1, \ldots,\Pi_N \}$ be the set of nontrivial projectors on $C$. Let $\Pi_i = \sum_{j=1}^{r_i} \ket{a_{ij}} \bra{a_{ij}}$, where $r_i \leq 2^k$ is the rank of $\Pi_i$ and $\ket{a_{ij}} \in (\mathbb{C}^2)^{\otimes k}$. (Computing the $\ket{a_{ij}}$ comes down to Gaussian elimination, which can be done in constant time, since arithmetic operations on the coefficients can be performed in constant time.) Let $\ket{\phi} = \ket{\phi_1} \otimes \cdots \otimes \ket{\phi_c}$ be a unknown product state defined on $C$. We set
\[
\ket{\phi_i} = x_{i0} \ket{0} + x_{i1} \ket{1}, \quad i\in\{1,2,\ldots,c\}.
\]
Let $s_i \in \sigma_k(C)$ be the set of qubits on which $\Pi_i$ acts. For each $s_i$, we write $\ket{\phi_{s_i}}:= \bigotimes_{v \in s_i} \ket{\phi_{v}}$. Then~\eqref{eq:satisfying} implies the following equations:
\be
\label{poly-system-1}
\langle a_{ij}| \phi_{s_i} \rangle = 0, \quad i\in\{1,2,\ldots,N\}, j \in\{1,2,\ldots,r_i\}.
\ee
We also have that $|x_{i0}|^2+|x_{i1}|^2 = 1$. As we can rescale \eqref{poly-system-1} arbitrarily, we can relax this to $|x_{i0}|^2+|x_{i1}|^2 \neq 0$; equivalently, $x_{ij}\neq 0$ for at least one $j$. To ensure this we introduce some new variables $\{y_{ij}\}$ with the same index sets as the $\{x_{ij}\}$, and some new equations:
\be
\label{poly-system-2}
(x_{i0} y_{i0} - 1) (x_{i1} y_{i1} - 1) = 0, \quad i \in \{1,\ldots,c\}.
\ee
By introducing a new variable $t$ and viewing $x_{ij}=x_{ij}/t$, $y_{ij}=y_{ij}/t$, the above inhomogeneous polynomial system becomes an equivalent homogeneous polynomial system, where the coefficients of the homogeneous polynomials lie in the field $\mathbb{F}$.
In total we have $ c + \sum_{i=1}^N r_i$ equations; each equation has degree at most $k$ or $4$. We want to determine if there is a solution for this polynomial system over $\mathbb{C}$. 
It is known that this polynomial system has no zeros in $\mathbb{C}$ if and only if its Gr\"obner basis is $\{1\}$. This is Hilbert's Nullstellensatz, which is an EXPSPACE-complete problem~\cite[Chapter 21.7]{von2013modern}. By~\cite[Proposition 1]{BARDET201549}, an upper bound on the time complexity in this case is
\beas
&& D(c + \sum_{i=1}^N r_i)  \binom{4c + D }{D}^\omega 
\eeas
where $D\leq 1 + 3c + (k-1) \sum_{i=1}^N r_i $ is the maximal degree of the elements in the reduced Gr\"obner basis (which is obtained by \cite[Equation (1)]{BARDET201549}), and $\omega<2.373$ is the exponent of matrix multiplication. 
The algorithm (e.g., see \cite{lazard1983grobner} and \cite[Algorithm matrix-F5]{BARDET201549}) for Gr\"obner basis involves Gaussian elimination of some matrices, which are constructed from the coefficients of the polynomials in a certain way. Again, this Gaussian elimination comes down to arithmetic operations over the field $\mathbb{F}$, which can be done in constant time.

We can therefore straightforwardly check satisfiability of an instance $([n],\{\Pi_s\})$ of quantum $k$-SAT with the testing promise. Since we are only interested in scaling w.r.t. $n$, we may assume that $n$ is large enough that (by Theorem~\ref{thm:localsat}) if the instance is satisfiable then it will be locally satisfiable by a product state on $c(k,\epsilon)$-subsets with probability $p > 0.75$, and (by Theorem~\ref{thm:localunsat}) if the instance is $\epsilon$-far from satisfiable by a product state then it will be locally unsatisfiable by a product state on $c(k,\epsilon)$-subsets with probability $p>0.75$.

We choose subsets $\{C_{i} \in \sigma_{c(k,\epsilon)}([n])\}_{i=1}^m$ at random and check whether the instance is locally satisfiable by a product state on these subsets. If the majority are locally satisfiable by a product state we conclude that the instance is satisfiable; otherwise we conclude that the instance is $\epsilon$-far from satisfiable by a product state. By the Chernoff bound for the binomial distribution the probability that the conclusion is incorrect tends exponentially to zero as $m$ increases.
\end{proof}

\subsection{Data availability statement}

Data sharing not applicable to this article as no datasets were generated or analysed during the current study.

\subsection{Acknowledgements}
We thank Niel de Beaudrap and Aram Harrow for useful discussions. We also thank our anonymous referees whose suggestions and comments have greatly improved the presentation of these results. This project has received funding from the European Research Council (ERC) under the European Union’s Horizon 2020 research and innovation programme (grant agreement No. 817581). We acknowledge support from EPSRC grant EP/T001062/1.

\section{Proof of Theorem~\ref{thm:localsat}} \label{sec:localsat}

In this proof we will make use of some combinatorial lemmas, some of which are well-known. For the reader's convenience we prove all these lemmas, but to avoid a long digression we have relegated the proofs to Appendix~\ref{app:combinatorics}. 
We will present the lemmas as they arise in the order of the proof. 

We begin by recalling several facts about entangled subspaces. Recall that a subspace $L \subset \mathbb{C}^2 \otimes \mathbb{C}^m$ is \emph{completely entangled} if there is no product state in $L$. 

\begin{lemma}[{\cite[Lemmas 1 and 2]{Augusiak2011}, also follows from \cite[Prop. 1.4]{Parthasarathy2004}}]\label{lem:completelyent}
Let $L\subset \mathbb{C}^2 \otimes \mathbb{C}^m$ be a completely entangled subspace, for any $m \geq 2$. Then the orthogonal complement $L^{\perp} \subset \mathbb{C}^2 \otimes \mathbb{C}^m$ contains a product state. 
\end{lemma}
\begin{corollary}\label{cor:completelyent}
Let $m \in \mathbb{N}$, and let $V \subset (\mathbb{C}^2)^{\otimes m}$ be a subspace of dimension $\geq 2$. Let $L \subset \mathbb{C}^2 \otimes (\mathbb{C}^2)^{\otimes m}$ be a subspace such that:
\begin{enumerate} \item $\mathbb{C}^2 \otimes V^{\perp} \subset L$.
\item $(\mathbb{C}^2 \otimes V) \cap L$ is a completely entangled subspace of $\mathbb{C}^2 \otimes V$.
\end{enumerate}
Then $L^{\perp} \subset \mathbb{C}^2 \otimes \mathbb{C}^m$ contains a product state.
\end{corollary}
\begin{proof}
Since $\mathbb{C}^2 \otimes V^{\perp} \subset L$, we have $L = (\mathbb{C}^2 \otimes V^{\perp}) \oplus \left((\mathbb{C}^2 \otimes V^{\perp})^{\perp} \cap L \right) = (\mathbb{C}^2 \otimes V^{\perp}) \oplus \left( (\mathbb{C}^2 \otimes V) \cap L \right)$. Then $L^{\perp} = (\mathbb{C}^2 \otimes V) \cap  \left( (\mathbb{C}^2 \otimes V) \cap L \right)^{\perp} = \left( (\mathbb{C}^2 \otimes V) \cap L \right)^{\widetilde{\perp}}$, where $\widetilde{\perp}$ is the orthogonal complement in $\mathbb{C}^2 \otimes V$. But $(\mathbb{C}^2 \otimes V) \cap L$ is a completely entangled subspace of $\mathbb{C}^2 \otimes V$, so by Lemma~\ref{lem:completelyent} we know that $\left( (\mathbb{C}^2 \otimes V) \cap L \right)^{\widetilde{\perp}}$ contains a product state; therefore $L^{\perp}$ contains a product state. 
\end{proof}
\begin{lemma}\label{lem:suppsprod}
Let $\ket{\psi} \in (\mathbb{C}^2)^{\otimes n}$ be a state and let $x \in [n]$. Let $S_1,S_2 \subset ([n]-\{x\})$ be two subsets such that $S_1 \cap S_2 = \emptyset$ and $S_1 \sqcup S_2 = ([n]-\{x\})$.  Then at least one of the following is true:
\begin{enumerate}
\item\label{num:prodlem1} The subspace $\supp(\Tr_{S_2}(\ket{\psi}\bra{\psi}))$ contains a product state $\ket{\phi_{x}} \otimes \ket{\phi_{S_1}}$. 
\item\label{num:prodlem2} The subspace $\supp(\Tr_{S_1}(\ket{\psi}\bra{\psi}))$ contains a product state $\ket{\phi_{x}} \otimes \ket{\phi_{S_2}}$. 
\end{enumerate}
\end{lemma}
\begin{proof}
To clarify our notation, we first recall the definition of supports and kernels. By definition, for any subset $S \subset [n]$, we have $\supp(\Tr_{\overline{S}}(\ket{\psi}\bra{\psi})) = \ker(\Tr_{\overline{S}}(\ket{\psi}\bra{\psi}))^{\perp}$, where the orthogonal complement is taken in the space $(\mathbb{C}^2)^{\otimes |S|}$. If we write $(\mathbb{C}^2)^{\otimes n} = (\mathbb{C}^2)^{\otimes |S|} \otimes (\mathbb{C}^2)^{\otimes (n-|S|)}$, where the first factor corresponds to those qubits in the set $S$, then: 
\bea
\ker(\Tr_{\overline{S}}(\ket{\psi}\bra{\psi})) &=& \textrm{span}\{\ket{v}\in (\mathbb{C}^2)^{\otimes |S|}~:~ (\bra{v} \otimes \mathbbm{1}) \ket{\psi} = 0 \}, \label{eq:ker} \\
\supp(\Tr_{S}(\ket{\psi}\bra{\psi}))) &=& \textrm{span}\{(\bra{v} \otimes \mathbbm{1}) \ket{\psi} ~:~\ket{v} \in (\mathbb{C}^2)^{|S|}\}. \label{eq:supp}
\eea
We now observe that if any of (i) $\dim(\supp(\Tr_{\overline{x}}(\ket{\psi}\bra{\psi})))$, (ii) $\dim(\supp(\Tr_{\overline{S_1}}(\ket{\psi}\bra{\psi})))$ or (iii) $\dim(\supp(\Tr_{\overline{S_2}}(\ket{\psi}\bra{\psi})))$ equal 1, then at least one of
the two claims in this lemma
is true. Indeed, for (i) observe that if $\dim(\supp(\Tr_{\overline{x}}(\ket{\psi}\bra{\psi}))) = 1$, then $\ket{\psi} = \ket{\psi_{x}} \otimes \ket{\psi_{([n]-\{x\})}}$. To see this, observe that there exists $\ket{w} \in \mathbb{C}^2$ such that $(\bra{w} \otimes \mathbbm{1}) \ket{\psi} = 0$; it follows by the Schmidt decomposition that $\ket{\psi}$ is a product state. The argument for (ii) and (iii) is similar. 

We therefore now assume that those three dimensions are greater than 1. We consider the subspace 
\begin{align*}
\ker(\Tr_{\overline{\{x\} \sqcup S_1}}(\ket{\psi}\bra{\psi})) \cap \big(\supp(\Tr_{\overline{\{x\}}}(\ket{\psi}\bra{\psi})) &\otimes \supp(\Tr_{\overline{S_1}}(\ket{\psi}\bra{\psi}))\big) \\
& 
\hspace{-1cm}
\subseteq \supp(\Tr_{\overline{\{x\}}}(\ket{\psi}\bra{\psi})) \otimes \supp(\Tr_{\overline{S_1}}(\ket{\psi}\bra{\psi})).
\end{align*}
We have a dichotomy:
\begin{enumerate}[label=(\alph*)]
    \item \label{num:isceclem=1}The subspace $\ker(\Tr_{\overline{\{x\} \sqcup S_1}}(\ket{\psi}\bra{\psi})) \cap \big(\supp(\Tr_{\overline{\{x\}}}(\ket{\psi}\bra{\psi})) \otimes \supp(\Tr_{\overline{S_1}}(\ket{\psi}\bra{\psi}))\big)$ is a completely entangled subspace of $\supp(\Tr_{\overline{\{x\}}}(\ket{\psi}\bra{\psi})) \otimes \supp(\Tr_{\overline{S_1}}(\ket{\psi}\bra{\psi}))$.
    \item \label{num:isceclem=2}The subspace $\ker(\Tr_{\overline{\{x\} \sqcup S_1}}(\ket{\psi}\bra{\psi})) \cap \big(\supp(\Tr_{\overline{\{x\}}}(\ket{\psi}\bra{\psi})) \otimes \supp(\Tr_{\overline{S_1}}(\ket{\psi}\bra{\psi}))\big)$ contains a product state. 
\end{enumerate}
In case~\ref{num:isceclem=1}, we are precisely in the situation of Corollary~\ref{cor:completelyent}, where  $V=\supp(\Tr_{\overline{S_1}}(\ket{\psi}\bra{\psi}))$ and $L = \ker(\Tr_{\overline{\{x\} \sqcup S_1}}(\ket{\psi}\bra{\psi}))$; therefore claim~\ref{num:prodlem1} is true. On the other hand, in case~\ref{num:isceclem=2} there exists a product state $\ket{v} \otimes \ket{w} \in \supp(\Tr_{\overline{\{x\}}}(\ket{\psi}\bra{\psi})) \otimes \supp(\Tr_{\overline{S_1}}(\ket{\psi}\bra{\psi}))$ such that $(\bra{v} \otimes \bra{w} \otimes \mathbbm{1}) \ket{\psi} = 0$. Consider the state $\ket{\psi'}:= (\mathbbm{1} \otimes \bra{w} \otimes \mathbbm{1}) \ket{\psi} \in \supp(\Tr_{S_1}(\ket{\psi}\bra{\psi}))$. Since $(\bra{v} \otimes \mathbbm{1}) \ket{\psi'}=0$, we have that $\dim(\ker(\Tr_{\overline{\{x\}}}(\ket{\psi'}\bra{\psi'}))) = 1$; by the Schmidt decomposition this implies that $\ket{\psi'} = \ket{\psi_{x}'} \otimes \ket{\psi_{S_2}'}$, and so claim~\ref{num:prodlem2} is true.
\end{proof}
\noindent
We can now prove the theorem. To warm up, we prove the case $c=2$ first. 
\begin{proposition}\label{prop:c=2}
Let $([n],\{\Pi_s\})$ be a satisfiable instance of quantum $k$-SAT, and let $x_1 \in [n]$ be any qubit. Then there is at most one qubit $x_2 \in ([n]-\{x_1\})$ such that the instance is not locally satisfiable by a product state at $\{x_1,x_2\}$.

Equivalently, let $\ket{\psi} \in (\mathbb{C}^{2})^{\otimes n}$ be any state, and let $x_1 \in [n]$ be any qubit. Then there is at most one qubit $x_2 \in ([n] - \{x_1\})$ such that the subspace $\supp(\Tr_{\overline{\{x_1,x_2\}}}(\ket{\psi}\bra{\psi})) \subseteq (\mathbb{C}^2)^{\otimes 2}$ does not contain a product state.
\end{proposition}
\begin{proof}
The equivalence of the two statements will be shown in the first paragraph of the proof of Theorem~\ref{thm:localsat} below. We therefore need only prove the second statement. Let $x_2 \in ([n]-\{x_1\})$ be any qubit. We now use Lemma~\ref{lem:suppsprod}; in the notation of that lemma, let $x:= x_1$, $S_1 := \{x_2\}$ and $S_2 := [n]-\{x_1,x_2\}$. Then either $\supp(\Tr_{\overline{\{x_1,x_2\}}}(\ket{\psi}\bra{\psi})))$ contains a product state, or  $\supp(\Tr_{\{x_2\}}(\ket{\psi}\bra{\psi})))$ contains a state which is a product over $x_1\sqcup S_2$. In the latter case it clearly follows from the definition of the support that there is a product state in $\supp(\Tr_{\overline{\{x_1 \sqcup x'\}}}(\ket{\psi}\bra{\psi})))$  for any $x' \in S_2$.
\end{proof}
\noindent
We now move onto the case $c \geq 3$.
\begin{theorem*}[Restatement of Theorem~\ref{thm:localsat}]
Let $([n],\{\Pi_s\})$ be a satisfiable instance of quantum $k$-SAT, and let $c \in \mathbb{N}$ be fixed, where $c \geq 3$. Let $C \in \sigma_c([n])$ be a subset chosen uniformly at random. The probability that the instance is locally satisfiable by a product state at $C$ is greater than $p \in (0,1)$ whenever $n > \Psi(p,c)$, where 
$$\Psi(p,c) := \max\left( \frac{c^8(2^{3c+20})}{(1-\sqrt{p})^5}~,~\frac{2^{c+2}(c+1)}{-\ln(p)}\right).$$

Equivalently, let $\ket{\psi} \in (\mathbb{C}^{2})^{\otimes n}$ be any state and let $C \in \sigma_c([n])$ be a subset chosen uniformly at random. The probability that the subspace $\supp(\Tr_{\overline{C}}(\ket{\psi}\bra{\psi})) \subseteq (\mathbb{C}^2)^{\otimes c}$ contains a product state is greater than $p \in (0,1)$ whenever $n > \Psi(p,c)$.
\end{theorem*}

\begin{remark}\label{rem:entanglement}
Before giving the proof, we remark that the following weaker fact follows quite straightforwardly from known results: 
\begin{quote} Let $([n],\{\Pi_s\})$ be a satisfiable instance of quantum $k$-SAT, let $c \in \mathbb{N}$ be some constant, and let $\epsilon > 0$. Let $C \in \sigma_c([n])$ be a subset chosen uniformly at random. The probability that there exists a product state $\ket{\phi} \in (\mathbb{C}^2)^{\otimes c}$ such that $\sum_{s \in \sigma_k(C)}\bra{\phi}\Pi_s \ket{\phi} < \epsilon$ is greater than $p$ whenever $p > O(n^{k-c-1/4}){c \choose k}/\epsilon$.
\end{quote}
To prove this, one can use~\cite[Equation (91)]{Brandao2016}, which implies that there is a product state $\ket{\phi} \in (\mathbb{C}^2)^{\otimes n}$ such that $\sum_{s \in \sigma_k([n])} \bra{\phi} \Pi_s \ket{\phi} \leq O(n^{k-1/4})$. We can write 
$$
\sum_{s \in \sigma_k([n])} \bra{\phi} \Pi_s \ket{\phi} =  \sum_{C \in \sigma_c([n])} \sum_{s \in \sigma_k(C)}  \frac{1}{{c \choose k}} \bra{\phi} \Pi_s \ket{\phi} = \sum_{C \in \sigma_c([n])} S_C,
$$
where $S_C := \sum_{s \in \sigma_k(C)}{c \choose k}^{-1}\bra{\phi} \Pi_s \ket{\phi}$. Since there are $O(n^c)$ $c$-subsets, the expected value of $S_C$ for a randomly chosen $c$-subset is $O(n^{(k-c)-1/4})$; the result then follows by Markov's inequality. 

An alternative approach to proving this fact, which has some issues but nevertheless reveals something about the nature of the problem, is to use monogamy of the squashed entanglement~\cite[Theorem 8]{Koashi2004}. For this, let $\ket{\psi}$ be a state satisfying $([n],\{\Pi_s\})$. It follows straightforwardly from monogamy that, for a randomly chosen $C \in \sigma_{c}([n])$, the expected entanglement across every bipartition of $C$ is $O(1/n)$. Then by~\cite[Theorem in \S{}2]{Brandao2011}, $\Tr_{\overline{C}}(\ket{\psi}\bra{\psi})$ is close to a biseparable state across every bipartition of $C$. If we could conclude that this implied closeness to a fully separable state, the fact would follow.

Theorem~\ref{thm:localsat} is a stronger, exact analogue, where we demand exact local satisfiability by a product state rather than approximate satisfiability. We were unable to derive Theorem~\ref{thm:localsat} from the non-exact result, and the proof we are about to provide follows a different approach, which does not make use of any continuous entanglement measure.
\end{remark}

\begin{proof}[Proof of Theorem~\ref{thm:localsat}]
We first observe the equivalence of the two statements in this theorem. If $([n],\{\Pi_s\})$ is satisfiable, let $\ket{\psi} \in (\mathbb{C}^2)^{\otimes n}$ be a satisfying state; clearly any state $\ket{\psi_{C}} \in \supp(\Tr_{\overline{C}}(\ket{\psi}\bra{\psi}))$ is a local solution to $([n],\{\Pi_s\})$ at $C$, so the second statement implies the first. In the other direction, let $\ket{\psi} \in (\mathbb{C}^2)^{\otimes n}$ be a state; then for all $C \in \sigma_c([n])$ let $\Pi_{C}$ be the projector onto $\ker(\Tr_{\overline{C}}(\ket{\psi}\bra{\psi}))$, and consider the instance $([n],\{\Pi_C\}_{C \in \sigma_c([n])})$ of quantum $c$-SAT, which is of course satisfied by the state $\ket{\psi}$.

Given this equivalence it is sufficient to show that for any state $\ket{\psi} \in (\mathbb{C}^2)^{\otimes n}$ and a randomly chosen subset $C \in \sigma_c([n])$, there exists a product state in $\supp(\Tr_{\overline{C}}(\ket{\psi}\bra{\psi}))$ with high probability.

Our idea to prove this is as follows. We will choose qubits $\{x_1,\dots,x_c\} =: C$ one by one uniformly at random. As we pick these qubits we will construct a sequence of hypergraphs $G_1,\dots,G_{c-1}$ whose vertices are the qubits which have not yet been picked; that is, the vertex set of $G_i$ is $[n]-\{x_1,\dots,x_{i}\}$. These hypergraphs are uniform, i.e. all their hyperedges (which for simplicity we will just call `edges') have the same cardinality. The significance of these hypergraphs is that, for any edge $E$ of the hypergraph $G_i$, we know that there is a state in the subspace 
$$
\supp(\Tr_{\overline{\{x_1,\dots,x_i\} \sqcup E}}(\ket{\psi}\bra{\psi})) \subset (\mathbb{C}^2)^{\otimes (i + |E|)}
$$
which is a product between the qubits $\{x_1,\dots,x_{i}\}$ that have already been picked and the qubits in edge $E$. To be precise:
\begin{quote}
Let $E \subset [n]-\{x_1,\dots,x_{i}\}$ be an edge of the hypergraph $G_i$. Then the subspace
$$
\supp(\Tr_{\overline{\{x_1,\dots,x_i\} \sqcup E}}(\ket{\psi}\bra{\psi})) \subset (\mathbb{C}^2)^{\otimes (i + |E|)}
$$
contains a state
\begin{align}\label{eq:giproperty}
\ket{\psi_1} \otimes \dots \otimes \ket{\psi_i} \otimes \ket{\psi_E}
\end{align}
where $\ket{\psi_1},\dots,\ket{\psi_i}$ are states of the qubits $x_1,\dots,x_i$ respectively, and $\ket{\psi_E}$ is a state of the qubits in $E$ (which may not itself be a product state).
\end{quote}
We will continue picking qubits at random and constructing these hypergraphs until we have picked the qubits $\{x_1,\dots,x_{c-1}\}$ and constructed the hypergraph $G_{c-1}$ on the remaining qubits. This hypergraph has the following property, which is a special case of the property just given for general $G_{i}$:
\begin{quote}
Let $E \subset [n]-\{x_1,\dots,x_{c-1}\}$ be an edge of the hypergraph $G_{c-1}$. Then the subspace
$$
\supp(\Tr_{\overline{\{x_1,\dots,x_{c-1}\} \sqcup E}}(\ket{\psi}\bra{\psi})) \subset (\mathbb{C}^2)^{\otimes (c-1 + |E|)}
$$
contains a state
\begin{align}\label{eq:gc-1property}
\ket{\psi_1} \otimes \dots \otimes \ket{\psi_{c-1}} \otimes \ket{\psi_E}
\end{align}
where $\ket{\psi_1},\dots,\ket{\psi_{c-1}}$ are states of the qubits $x_1,\dots,x_{c-1}$ respectively, and $\ket{\psi_E}$ is a state of the qubits in $E$ (which may not itself be a product state).
\end{quote}
Since there is a state $
\ket{\psi_1} \otimes \dots \otimes \ket{\psi_{c-1}} \otimes \ket{\psi_E}
$ in $\supp(\Tr_{\overline{\{x_1,\dots,x_{c-1}\} \sqcup E}}(\ket{\psi}\bra{\psi}))$ then it follows straightforwardly from the definition of the support given in~\eqref{eq:supp} that, if the randomly chosen final qubit $x_{c}$ lies in $E$, there will be a product state in $\supp(\Tr_{\overline{C}}(\ket{\psi}\bra{\psi}))$. Indeed, by definition of the support~\eqref{eq:supp} there exists some vector $\ket{v} \in (\mathbb{C}^2)^{\otimes (n-(c-1+|E|))}$ such that 
$$
(\mathbbm{1} \otimes \bra{v})\ket{\psi} = \ket{\psi_1} \otimes \dots \otimes \ket{\psi_{c-1}} \otimes \ket{\psi_E}
$$
Now let $\ket{w} \in (\mathbb{C}^2)^{|E|-1}$ be any vector in the Hilbert space of the qubits in $(E-\{x_{c}\})$ such that 
$$
||(\mathbbm{1}\otimes \bra{w})\ket{\psi_E}|| = 1.
$$
Then we see that 
$$
(\mathbbm{1} \otimes \bra{w} \otimes \bra{v})\ket{\psi} = \ket{\psi_1} \otimes \dots \otimes \ket{\psi_{c-1}} \otimes \ket{\psi_c},
$$
where $\ket{\psi_c} = (\mathbbm{1}\otimes \bra{w})\ket{\psi_E}$; and therefore, by definition~\eqref{eq:supp} there is a product state in $$\supp(\Tr_{\overline{C}}(\ket{\psi}\bra{\psi})).$$
Note that the final qubit $x_c$, like all the other qubits in $C$, is randomly chosen and we therefore cannot stipulate that it is contained in $E$.

The probability that there is a product state in $\supp(\Tr_{\overline{C}}(\ket{\psi}\bra{\psi}))$ is therefore lower bounded by the probability that the final qubit we pick lies in an edge of $G_{c-1}$.
We will lower bound this latter probability; our lower bound depends on the edge size and edge density of the uniform hypergraph $G_{c-1}$ (recall that the \emph{edge density} $\theta \in [0,1]$ of a $k$-uniform hypergraph on $n$ vertices is the proportion of the $\binom{n}{k}$ possible edges which are actually edges of the hypergraph). In this context, we have our first two combinatorial lemmas. 
\begin{lemma}\label{lem:finalvertex}
    Let $G$ be a $k$-uniform hypergraph on $n$ vertices with edge density $\theta \in [0,1]$, where $\theta^{1/k}n \geq k$. If we pick a vertex at random, the probability that it lies within an edge of $G$ is at least $\theta^{1/k}$.
\end{lemma}
\noindent
These lemmas are proved in Appendix~\ref{app:combinatorics}. Given Lemma~\ref{lem:finalvertex}, we can lower bound the probability that the qubit $x_{c}$ lies within an edge of $G_{c-1}$ by lower bounding the edge density and edge size of $G_{c-1}$. 

\paragraph{Construction of the hypergraph $G_1$.}We will begin by defining the hypergraph $G_1$. 
We pick the first qubit $x_{1}$ at random. There are two possibilities:
\begin{enumerate}[label=(\alph*)]
    \item \label{num:nodd} $n-1$ is even.
    \item \label{num:neven} $n-1$ is odd.
\end{enumerate}
In case~\ref{num:nodd}, we consider all the subsets of the set $([n]-\{x_1\})$ of size $(n-1)/2$; recall that our notation for the set of all such subsets is $\sigma_{(n-1)/2}([n]-\{x_1\})$. Let $X \in \sigma_{(n-1)/2}([n]-\{x_1\})$. Then, by Lemma~\ref{lem:suppsprod}, we know that at least one of the following holds:
\begin{enumerate}[label=(\roman*)]
\item\label{num:prodb1} The subspace $\supp(\Tr_{\overline{\{x_1\} \sqcup X}}(\ket{\psi}\bra{\psi}))$ contains a state $\ket{\psi_1} \otimes \ket{\psi_X}$. 
\item\label{num:prodb2} The subspace $\supp(\Tr_{\overline{\{x_1\} \sqcup \overline{X}}}(\ket{\psi}\bra{\psi}))$ contains a state $\ket{\psi_1} \otimes \ket{\psi_{\overline{X}}}$. 
\end{enumerate}
We therefore define an $(n-1)/2$-uniform hypergraph $G_1$ on the set $([n]-\{x_1\})$ as follows. For each $X \in \sigma_{(n-1)/2}([n]-\{x_1\})$, if $\supp(\Tr_{\overline{\{x_1\} \sqcup X}}(\ket{\psi}\bra{\psi}))$ contains a state $\ket{\psi_1} \otimes \ket{\psi_X}$, then add the subset $X$ as an edge to the hypergraph $G_1$. On the other hand, if $\supp(\Tr_{\overline{\{x_1\} \sqcup \overline{X}}}(\ket{\psi}\bra{\psi}))$ contains a state $\ket{\psi_1} \otimes \ket{\psi_{\overline{X}}}$, then add the subset $\overline{X}$ as an edge to the hypergraph $G_1$. By Lemma~\ref{lem:suppsprod}, the edge density $\theta_1$ of $G_1$ has the lower bound $\theta_1 \geq 1/2$. By definition, this hypergraph $G_1$ satisfies the desired property~\eqref{eq:giproperty}.

That definition of $G_1$ depended on the fact that $(n-1)$ was even. In case~\ref{num:neven}, we have to make a slightly different definition, which is only a technical annoyance. For any subset $X \in \sigma_{n/2-1}([n]-\{x_1\})$, at least one of the following holds, by Lemma~\ref{lem:suppsprod}:
\begin{enumerate}[label=(\roman*)]
\item\label{num:prodx} The subspace $\supp(\Tr_{\overline{\{x_1\} \sqcup X}}(\ket{\psi}\bra{\psi}))$ contains a state $\ket{\psi_1} \otimes \ket{\psi_X}$. 
\item The subspace $\supp(\Tr_{\overline{\{x_1\} \sqcup \overline{X}}}(\ket{\psi}\bra{\psi}))$ contains a state $\ket{\psi_1} \otimes \ket{\psi_{\overline{X}}}$. 
\end{enumerate}
We can therefore define an $(n/2-1)$-uniform hypergraph $G_1$ on the set $([n]-\{x_1\})$ as follows. For every $X \in \sigma_{n/2-1}([n]-\{x_1\})$, if $\supp(\Tr_{\overline{\{x_1\} \sqcup X}}(\ket{\psi}\bra{\psi}))$ contains a state $\ket{\psi_1} \otimes \ket{\psi_X}$ then we add $X$ as an edge to $G_1$. On the other hand, if $\supp(\Tr_{\overline{\{x_1\} \sqcup \overline{X}}}(\ket{\psi}\bra{\psi}))$ contains a state $\ket{\psi_1} \otimes \ket{\psi_{\overline{X}}}$ then we add all the subsets in $\sigma_{n/2-1}(\overline{X})$ as edges to $G_1$. Using this construction, we again find that the edge density $\theta_1$ of $G_1$ has the lower bound $\theta_1 \geq 1/2$, by the following combinatorial lemma.
\begin{lemma}[Edge density for a certain construction]\label{lem:oddedges}
Let $n \in \mathbb{N}$ be odd. We construct an $(n-1)/2$-uniform hypergraph $G$ as follows. For every subset $X \in \sigma_{(n-1)/2}([n])$, we do precisely one of the following:
\begin{enumerate}
    \item\label{num:option1} Add $X$ as an edge to $G$.
    \item\label{num:option2} Add all the subsets in $\sigma_{(n-1)/2}(\overline{X})$ as edges to $G$.
\end{enumerate}
The edge density $\theta$ of a hypergraph $G$ constructed in this way satisfies $\theta \geq 1/2$.
\end{lemma}
\noindent
This lemma is proved in Appendix~\ref{app:combinatorics}.
Again, by construction the hypergraph $G_1$ satisfies the desired property~\eqref{eq:giproperty}.

So, in summary, we have constructed a hypergraph $G_1$ on the vertex set $([n]-\{x_1\})$ satisfying the desired property~\eqref{eq:giproperty}, with edge density $\theta_1 \geq 1/2$, and which is $k_1$-uniform, where $k_1=(n-1)/2$ (if $n-1$ is even) or $k_1=n/2-1$ (if $n-1$ is odd).

\paragraph{Construction of the hypergraph $G_2$.} Having constructed the hypergraph $G_1$, we now pick the next qubit $x_2 \in ([n]-\{x_1\})$ at random. We define a new hypergraph $W_2$ ($W$ for `working', since this graph is an intermediate step towards obtaining $G_2$) on the vertex set $([n]-\{x_1,x_2\})$. The edges of $W_2$ are precisely the edges of $G_1$ which contain the vertex $x_2$, with that vertex removed. Formally, the edge set of $W_2$ is:
$$\{(E -\{x_2\})~|~ E\textrm{ is an edge of $G_1$ and } x_2 \in E\}.$$
The construction of $W_2$ from $G_1$ is illustrated in Figure~\ref{fig:w2}.
\begin{figure}
    \centering
    \includegraphics[valign=c,width=0.25\linewidth]{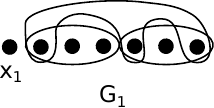}
    ~~~~~~~~~\includegraphics[valign=c,width=0.25\linewidth]{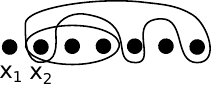}
    ~~~~~~~~~\includegraphics[valign=c,width=0.25\linewidth]{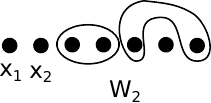}
    \caption{The construction of the hypergraph $W_2$ from the hypergraph $G_1$. The figures should be read from left to right. The first figure shows the graph $G_1$. In the second figure, the vertex $x_2$ is picked and all edges not containing that vertex are discarded. In the third figure, the vertex $x_2$ is removed from the remaining edges to obtain the hypergraph $W_2$ on $[n]-\{x_1,x_2\}$.}
    \label{fig:w2}
\end{figure}
This hypergraph $W_2$ is $l_2$-uniform, where $l_2 = ((n-3)/2)$ if $n-1$ was even, and $l_2 = (n/2-2)$ if $n-1$ was odd. To obtain a lower bound on its edge density, we use the following combinatorial lemma. Recall that the \emph{degree} of a vertex in a hypergraph is the number of edges containing it.
\begin{lemma}[Degree of a random vertex]\label{lem:randomvertex}
Let $G$ be a $k$-uniform hypergraph on $n$ vertices with edge density $\theta$, where $k = \beta n$ for some $\beta \in [0,1]$. Let $N$ be the degree of a vertex picked uniformly at random, and let $\tau:= N/{n-1 \choose k-1}$. Then, for any $\alpha \in(0,1)$:
\begin{align*}
\Pr\left[|\theta - \tau| \geq 
\Delta(\alpha,\beta,n)\right] \leq \alpha,
\end{align*}
where
\begin{align}\label{eq:deltadef}
\Delta(\alpha,\beta,n) =  \frac{5\exp\left(\frac{1}{12\alpha(1-\alpha)n}+\frac{1}{12\beta(1-\beta)n}+4\right)}{\alpha^2\sqrt{2\pi \alpha \beta (1-\alpha) (1-\beta) n}\left(1-\frac{2}{\beta n}\right)^{5/2}\left(1-\frac{2}{(1-\alpha)(1-\beta)n}\right)^{1/2}}.
\end{align}
\end{lemma}
\noindent
Note that when $\alpha,\beta$ are small constants, then $\Delta(\alpha,\beta,n)$ scales like $1/\sqrt{n}$.
The proof of this lemma can be found in Appendix~\ref{app:combinatorics}. 

Looking back at the construction of the hypergraph $W_2$, we see that the number of edges of $W_2$ is precisely the number of edges of $G_1$ which contained the randomly chosen vertex $x_2$; that is, the degree of the vertex $x_2$ in $G_1$. By Lemma~\ref{lem:randomvertex}, then, with probability greater than or equal to $(1-\alpha)$ the edge density of $W_2$ is lower bounded by $1/2-\Delta(\alpha,k_1/(n-1),n-1)$. For the moment, we will not fix a value for $\alpha \in (0,1)$; it is just a constant that we will fix later in order to obtain the theorem.

Now we will define the graph $G_2$. First we observe that, by the property~\eqref{eq:giproperty} obeyed by $G_1$, it follows that
for every edge $E$ in $W_2$ there is a state $\ket{\psi_1} \otimes \ket{\psi_{\{x_2\} \sqcup E}}$ in $\supp(\Tr_{\overline{\{x_1,x_2\}\sqcup E}}(\ket{\psi}\bra{\psi}))$.
If $l_2$ is even, we again consider all $(l_2/2)$-subsets of the edges of $W_2$: for every edge $E$ of $W_2$ and every $X \in \sigma_{l_2/2}(E)$, then at least one of the following holds, by applying Lemma~\ref{lem:suppsprod} to the state $\ket{\psi_{\{x_2\} \sqcup E}}$:
\begin{enumerate}[label=(\roman*)]
    \item\label{num2:prodx} The subspace $\supp(\Tr_{\overline{\{x_1,x_2\} \sqcup X}}(\ket{\psi}\bra{\psi})$ contains a state $\ket{\psi_{1}} \otimes \ket{\psi_{2}} \otimes \ket{\psi_X}$.
    \item The subspace $\supp(\Tr_{\overline{\{x_1,x_2\} \sqcup \overline{X}}}(\ket{\psi}\bra{\psi})$  contains a state $\ket{\psi_{1}} \otimes \ket{\psi_{2}} \otimes \ket{\psi_{\overline{X}}}$.
\end{enumerate}
We then define the $(l_2/2)$-uniform hypergraph $G_2$ as follows. For every edge $E$ in $W_2$, and every $X \in \sigma_{l_2/2}(E)$, if $\supp(\Tr_{\overline{\{x_1,x_2\} \sqcup X}}(\ket{\psi} \bra{\psi})$ contains a state $\ket{\psi_{1}} \otimes \ket{\psi_{2}} \otimes \ket{\psi_X}$ then we add $X$ as an edge to $G_2$. On the other hand, if $\supp(\Tr_{\overline{\{x_1,x_2\} \sqcup \overline{X}}}(\ket{\psi}\bra{\psi})$ contains a state $\ket{\psi_{1}} \otimes \ket{\psi_{2}} \otimes \ket{\psi_{\overline{X}}}$ then we add $\overline{X}$ as an edge to $G_2$.

On the other hand, if $l_2$ is odd then we define the $(l_2-1)/2$-uniform hypergraph $G_2$ as follows. For every edge $E$ of $W_2$ and every $X \in \sigma_{(l_2-1)/2}(E)$, at least one of the following holds, by applying Lemma~\ref{lem:suppsprod} to $\ket{\psi_{\{x_2\} \sqcup E}}$:
\begin{enumerate}[label=(\roman*)]
    \item The subspace $\supp(\Tr_{\overline{\{x_1,x_2\} \sqcup X}}(\ket{\psi}\bra{\psi})$ contains a state $\ket{\psi_{1}} \otimes \ket{\psi_{2}} \otimes \ket{\psi_X}$.
    \item The subspace $\supp(\Tr_{\overline{\{x_1,x_2\} \sqcup \overline{X}}}(\ket{\psi}\bra{\psi})$  contains a state $\ket{\psi_{1}} \otimes \ket{\psi_{2}} \otimes \ket{\psi_{\overline{X}}}$.
\end{enumerate}
For every edge $E$ in $W_2$ and every $X \in \sigma_{(l_2-1)/2}(E)$, if 
$\supp(\Tr_{\overline{\{x_1,x_2\} \sqcup X}}(\ket{\psi} \bra{\psi})$  contains a state $\ket{\psi_{1}} \otimes \ket{\psi_{2}} \otimes \ket{\psi_X}$ then we add $X$ as an edge to $G_2$. On the other hand, if $\supp(\Tr_{\overline{\{x_1,x_2\} \sqcup \overline{X}}}(\ket{\psi}\bra{\psi})$ contains a state $\ket{\psi_{1}} \otimes \ket{\psi_{2}} \otimes \ket{\psi_{\overline{X}}}$ then we add all the subsets in $\sigma_{(l_2-1)/2}(\overline{X})$ as edges to $G_2$. 

By construction, the graph $G_2$ obeys the desired property~\eqref{eq:giproperty}. 
We now need to obtain a lower bound on the edge density of $G_2$. For this we use the notion of the \emph{partial shadow} of a hypergraph. We first recall the notion of the shadow of a hypergraph, which is well known; see e.g.~\cite{Keevash2008}. 
\begin{definition}
Let $G$ be an $l$-uniform hypergraph on $n$ vertices, and let $k \leq l$. We define the \emph{$k$-shadow} $G|_k$ of $G$ to be the $k$-uniform hypergraph on $n$ vertices whose edges are precisely those $k$-subsets which are contained in an edge of $G$.
\end{definition}
\noindent
The notion of a partial shadow is less well-known, but has appeared in the combinatorial literature; see e.g.~\cite{Bollobas2015,Fitch2018}.
\begin{definition}
Let $G$ be an $l$-uniform hypergraph on $n$ vertices, let $k \leq l$, and let $\omega \in [0,1]$. We say that an $k$-uniform hypergraph $G|_{k,\omega}$ on $n$ vertices is a \emph{partial $(k,\omega)$-shadow} of $G$ if it can be obtained by the following construction. For each edge $E$ of $G$, choose a subset $K_{E} \subset \sigma_{k}(E)$ of size $|K_{E}| = \lfloor \omega {l \choose k} \rfloor$. The edge set of $G|_{k,\omega}$ is then defined as $\bigcup_{E \in G} K_{E}$.
\end{definition}
\noindent
Obviously there are many different partial $(k,\omega)$-shadows of an $l$-uniform hypergraph $G$; they depend on the choice made of the set $K_{E}$ for each edge $E \in G$. The following lemma gives a lower bound on the edge density of any partial shadow.
\begin{lemma}[Edge density of partial shadows]\label{lem:partialedgedensity}
Let $G$ be an $l$-uniform hypergraph with edge density $\theta$, and let $k\leq l$ and $\omega \in [0,1]$. The edge density $\theta|_{k,\omega}$ of any partial $(k,\omega)$-shadow of $G$ has the lower bound 
$$
\theta|_{k,\omega} \geq \omega \theta.
$$

\end{lemma}
\noindent 
This lemma is proved in Appendix~\ref{app:combinatorics}. Let us see how it allows us to lower bound the edge density of $G_2$. Recall that if the edge size $l_2$ of $W_2$ is even, $G_2$ is constructed by considering, for each edge $E$ of $W_2$, every subset in $X \in \sigma_{l_2/2}(E)$, and adding either $X$ or $\overline{X}$ (or both) as an edge of $G_2$. In this case the hypergraph $G_2$ therefore contains a partial $(l_2/2,1/2)$-shadow of $W_2$. By Lemma \ref{lem:randomvertex}, we know that the edge density of $W_2$ is lower bounded by $1/2-\Delta(\alpha,k_1/(n-1),n-1)$ with probability $(1-\alpha)$. By Lemma~\ref{lem:partialedgedensity}, therefore, with probability $(1-\alpha)$ the edge density $\theta_2$ of $G_2$ has the following lower bound:
\begin{align}\label{eq:lowerboundtheta2}
\theta_2 \geq \frac{1}{4} - \frac{\Delta(\alpha, k_1/(n-1),n-1)}{2} \geq \frac{1}{4} - \Delta(\alpha, k_1/(n-1),n-1).
\end{align}
Here for the second inequality (which will simplify the later analysis) we assume that $n$ is big enough that $\Delta(\alpha,k_1/(n-1),n-1)$ is positive. 

On the other hand, recall that if the edge size of $W_2$ is odd, $G_2$ is constructed by considering, for every edge $E$, every $(l_2-1)/2$-subset of that edge and adding either that subset or every $(l_2-1)/2$-subset in its complement to $G_2$. By Lemma~\ref{lem:oddedges}, we therefore add at least half of the $(l_2-1)/2$-subsets of each edge to $G_2$, and it follows that $G_2$ contains an $((l_2-1)/2,1/2)$-shadow of $W_2$. We can therefore again apply  Lemma~\ref{lem:partialedgedensity}, and we again obtain the lower bound~\eqref{eq:lowerboundtheta2} on the edge density of $G_2$ with probability $(1-\alpha)$. 

In summary, we have constructed a hypergraph $G_2$ on the vertex set $([n]-\{x_1,x_2\})$ satisfying the desired property~\eqref{eq:giproperty}, with edge density, which is $k_2$-uniform, where $k_2=l_2/2$ (if $l_2$ is even) or $k_2=(l_2-1)/2$ (if $l_2$ is odd), and whose edge density has the lower bound~\eqref{eq:lowerboundtheta2} with probability $(1-\alpha)$.

\paragraph{Construction of the hypergraphs $G_3,\dots,G_{c-1}$.} Now that the construction of $G_2$ from $G_1$ has been exemplified, it is straightforward to iterate the construction to obtain the hypergraphs $G_3,\dots,G_{c-1}$. We will briefly review the procedure to go from the hypergraph $G_i$ to the hypergraph $G_{i+1}$. First the qubit $x_{i+1}$ is picked at random. The hypergraph $W_{i+1}$ is defined on the vertex set $[n]-\{x_1,\dots,x_{i+1}\}$ by taking all edges of $G_i$ containing the vertex $x_{i+1}$, removing that vertex from those edges, and adding those edges (with the vertex $x_{i+1}$ removed) to the hypergraph $W_{i+1}$. The hypergraph $G_{i+1}$ is then defined as follows:
\begin{itemize}
    \item If the edge size $l_{i+1}$ of $W_{i+1}$ is even, then for every edge $E$ in $W_{i+1}$, and every $X \in \sigma_{l_{i+1}/2}(E)$, if $\supp(\Tr_{\overline{\{x_1,\dots,x_{i+1}\} \sqcup X}}(\ket{\psi} \bra{\psi}))$ contains a product state $\ket{\psi_{1}} \otimes \cdots \otimes \ket{\psi_{i+1}} \otimes \ket{\psi_X}$ then we add $X$ as an edge to $G_{i+1}$. On the other hand, if $\supp(\Tr_{\overline{\{x_1,\cdots,x_{i+1}\} \sqcup \overline{X}}}(\ket{\psi}\bra{\psi}))$ contains a product state $\ket{\psi_{1}} \otimes \cdots \otimes \ket{\psi_{i+1}} \otimes \ket{\psi_{\overline{X}}}$ then we add $\overline{X}$ as an edge to $G_{i+1}$.
\item If the edge size $l_{i+1}$ of $W_{i+1}$ is odd, then for every edge $E$ in $W_{i+1}$ and every $X \in \sigma_{(l_{i+1}-1)/2}(E)$, if 
$\supp(\Tr_{\overline{\{x_1,\dots,x_{i+1}\} \sqcup X}}(\ket{\psi} \bra{\psi}))$  contains a product state $\ket{\psi_{1}} \otimes \cdots \otimes \ket{\psi_{i+1}} \otimes \ket{\psi_X}$ then we add $X$ as an edge to $G_{i+1}$. On the other hand, if $\supp(\Tr_{\overline{\{x_1,\dots,x_{i+1}\} \sqcup \overline{X}}}(\ket{\psi}\bra{\psi}))$ contains a product state $\ket{\psi_{1}} \otimes  \cdots \otimes \ket{\psi_{i+1}} \otimes \ket{\psi_{\overline{X}}}$ then we add all the subsets in $\sigma_{(l_{i+1}-1)/2}(\overline{X})$ as edges to $G_{i+1}$. 
\end{itemize}
By definition, the hypergraph $G_{i+1}$ satisfies the desired property~\eqref{eq:giproperty}. By Lemma~\ref{lem:suppsprod} and Lemma~\ref{lem:oddedges}, using the same argument as we made for the transition from $G_{1}$ to $G_2$, we see that $G_{i+1}$ contains a partial $(l_{i+1}/2,1/2)$-shadow of $W_{i+1}$ (if the edge size of $W_{i+1}$ is even) and a partial $((l_{i+1}-1)/2,1/2)$-shadow of $W_{i+1}$ (if the edge size of $W_{i+1}$ is odd). This will allow us to lower bound the edge size and edge density of $G_{c-1}$ using Lemma~\ref{lem:randomvertex} and Lemma~\ref{lem:partialedgedensity}, as we did for $G_2$.

\paragraph{Lower bounds on the edge size and edge density of $G_{c-1}$.}
Let us first compute bounds on the edge size of the hypergraph $G_i$. From the construction we have outlined we see that, where the edge size of the hypergraph $G_i$ is $k_i$ and the edge size of the hypergraph $W_i$ is $l_i$, we have
$$
l_{i+1} = k_i-1
$$
and 
$$
k_{i+1} = l_{i+1}/2 = (k_i-1)/2 \mathrm{ ~~~or~~~ } k_{i+1} = l_{i+1}/2-1/2 = k_{i}/2 -1.
$$
The choice between the two values of $k_{i+1}$ depends on whether the edge size $l_{i+1}$ of $W_{i+1}$ is even or odd. It follows that an upper bound on the edge size is given by solving the linear recurrence relation $k_{i+1} = (k_{i}-1)/2$ with the initial condition $k_{0} = n$, and an upper bound is given by solving the linear recurrence relation $k_{i+1} = k_i/2 -1$ with the initial condition $k_0 = n$.
We thereby obtain the following bounds on $k_i, l_{i+1}$:
\begin{align}\nonumber
    n/2^i > n/2^i +1/2^i-1 \geq k_{i} &\geq n/2^i+1/2^{i-1}-2 > n/2^i-2, \\\label{eq:lkbounds}
    n/2^i -1 > n/2^i +1/2^i-2 \geq l_{i+1} &\geq n/2^i+1/2^{i-1}-3 > n/2^i - 3.
\end{align}
In particular, 
\begin{align}\label{eq:kc-1bound}
n/2^{c-1} > k_{c-1} > n/2^{c-1}-2. 
\end{align}

We now want to obtain a lower bound on the edge density of these hypergraphs. We can make the same argument as for the passage from $G_1$ to $G_2$: if $\theta_i$ is the edge density of $G_i$, then with probability at least $(1-\alpha)$ the edge density of $G_{i+1}$ has the following lower bound:
$$
\theta_{i+1} \geq \frac{\theta_{i}}{2} - \Delta(\alpha,k_{i}/(n-i),n-i).
$$
Given that the edge density $\theta_1$ of $G_1$ is at least $1/2$, and using the union bound for the probabilities, it follows that, with probability at least $(1-(c-2)\alpha)$,
\begin{align}\label{eq:c-1basicbound}
\theta_{c-1} \geq \frac{1}{2^{c-1}} - \Delta(\alpha, k_1/(n-1),n-1) - \dots - \Delta(\alpha,k_{c-2}/(n-(c-2)),n-(c-2)).
\end{align}
We have thus obtained our desired lower bounds on the edge size and edge density of $G_{c-1}$. We will now make numerous approximations to simplify our subsequent calculations. Throughout we will accumulate lower bounds on $n$ that we will bring together at the end of the proof. 

Firstly, we will simplify 
$\Delta(\alpha,k_i/(n-i),n-i)$; recall the definition~\eqref{eq:deltadef}. We observe that if the conditions
$$
n \geq \frac{1}{12}\left(\frac{1}{\alpha(1-\alpha)}+\frac{1}{\beta(1-\beta)}\right)+i 
$$
$$
n \geq \frac{2}{\beta(1-(1/2)^{2/5})}+i
$$
$$
n \geq \frac{8}{3(1-\alpha)(1-\beta)}+i
$$
are obeyed, then $$\Delta(\alpha,\beta,n-i) \leq \frac{20e^5}{\alpha^2\sqrt{2\pi\alpha(1-\alpha)\beta(1-\beta)(n-i)}}.$$
Now observe from~\eqref{eq:lkbounds} that, for 
\begin{align}\label{eq:ncondnew1}
    n \geq 2^{c+5},
\end{align} since $i \leq c$, and $i \geq 2$ whenever we apply Lemma~\ref{lem:randomvertex}:
\begin{align}\label{eq:kiinequal} 
k_i/(n-i) > n/(2^{i}(n-i)) - 2/(n-i)>1/2^{i} - 1/2^{c+3} > 1/2^{c+1}
\end{align}
\begin{align*}
k_i/(n-i) < n/(2^i(n-i)) < 1/2^{i-1} < 1/2
\end{align*}
We therefore conclude from what has just been said that if 
\begin{align}\label{eq:ncondnew21}
n \geq \frac{1}{12}\left(\frac{1}{\alpha(1-\alpha)}+2^{c+2}\right)+c 
\end{align}
\begin{align}
\label{eq:ncondnew22}
n \geq 2^{c+5} + c    
\end{align}
\begin{align}
\label{eq:ncondnew23}
n \geq \frac{2^{3}}{1-\alpha}+c,  
\end{align}
then 
\begin{align}\label{eq:deltatildedef}\Delta(\alpha,k_i/(n-i),n-i) \leq \widetilde{\Delta}(\alpha,n):= \frac{20e^5(2^{c+1})}{\alpha^2\sqrt{2\pi\alpha(1-\alpha)(2^{c+1}-1)(n-c)}}.\end{align}
(This can be seen by observing that $\beta(1-\beta)$ is a quadratic with maximum at $\beta = 1/2$; therefore a lower bound for $\beta = k_i/(n-i)< 1/2$ will give an lower bound for $\beta(1-\beta)$, and by~\eqref{eq:kiinequal} we have $k_i/(n-i) > 1/2^{c+1}$.) 
Applying these simplifications to~\eqref{eq:c-1basicbound}, we obtain
\begin{align}\label{eq:c-1edgedensity}
\Pr[\theta_{c-1} \geq \Phi] \geq 1-(c-2)\alpha \geq 1-c\alpha,
\end{align}
where 
\begin{align}\label{eq:defphi}
\Phi =  \frac{1}{2^{c}}-c\widetilde{\Delta}(\alpha,n)
\end{align}

\paragraph{Final computations.}  Now that we have lower bounds on the edge size and edge density of the hypergraph $G_{c-1}$, we can go back to Lemma~\ref{lem:finalvertex} and lower bound the probability that the final picked vertex $x_c$ lies in an edge of that hypergraph. 

By~\eqref{eq:c-1edgedensity} we know that the probability that the edge density of $G_{c-1}$ is greater than $\Phi$ is lower bounded by $1-c\alpha$. Given that the edge density of $G_{c-1}$ is greater than $\Phi$, we see from Lemma~\ref{lem:finalvertex} that the probability of a randomly picked vertex lying in an edge of $G_{c-1}$ is greater than $(\Phi)^{1/k_{c-1}}$, provided that \begin{align}\label{eq:binomfraccondition1}(\Phi)^{1/k_{c-1}}(n-(c-1)) \geq k_{c-1}.\end{align}
We will check this condition at the end of the proof, but for now let us assume that it holds. 

It follows that the probability $p$ that $\supp(\Tr_{\overline{C}}(\ket{\psi}\bra{\psi}))$ contains a product state has the following lower bound:
\begin{align}\label{eq:pinequality}
p \geq (1-c\alpha)(\Phi)^{1/k_{c-1}} \geq (1-c\alpha)(\Phi)^{2^{c-1}/(n-2^{c})}
\end{align}
In order to obtain the theorem we will suppose that some desired lower bound $\overline{p}$ for $p$ is given and choose the free variables $\alpha, n$ in order to satisfy the inequality  
\begin{align}\label{eq:poverlineequality}
(1-c\alpha)(\Phi)^{2^{c-1}/(n-2^{c})} \geq \overline{p}
\end{align}
while keeping $n$ as small as possible. We have already made and will continue to make numerous simplifying assumptions, so the value of $n$ we arrive at will not be optimal. 

Assuming that $1-c\alpha >0$, substituting in the definitions~\eqref{eq:deltatildedef}, \eqref{eq:defphi} of $\Phi$ and $\widetilde{\Delta}$ into ~\eqref{eq:poverlineequality} and rearranging, we obtain:
$$
\left(\frac{\overline{p}}{1-c\alpha}\right)^{(n-2^c)/2^{c-1}} \leq \frac{1}{2^{c}} - \frac{20e^5c(2^{c+1})}{\alpha^2\sqrt{2\pi\alpha(1-\alpha)(2^{c+1}-1)(n-c)}}
$$
As $n$ gets bigger, the quantity on the right of the inequality becomes larger. Let us therefore suppose that $n$ is large enough that 
\begin{align}\label{eq:nconditionnew3}
\frac{20e^5c(2^{c+1})}{\alpha^2\sqrt{2\pi\alpha(1-\alpha)(2^{c+1}-1)(n-c)}} \leq \frac{1}{2^{c+1}}.
\end{align}
This is implied by 
\begin{align}\label{eq:nconditionnew2}
n \geq \frac{c^2(2^{3c+20})}{\alpha^5(1-\alpha)}.
\end{align}
With this lower bound on the size of $n$, we obtain the simpler inequality
$$
\left(\frac{\overline{p}}{1-c\alpha}\right)^{(n-2^c)/2^{c-1}} \leq \frac{1}{2^{c+1}}.
$$
Let us now set \begin{align}\label{eq:alphadefnew}
\alpha = (1-\overline{p}^{1/2})/c.
\end{align}
This gives
$$
\overline{p}^{(n-2^{c})/2^c} \leq \frac{1}{2^{c+1}},
$$
which follows whenever 
\begin{align}\label{eq:ncondnew3}
n \geq 2^c\left(\frac{(c+1)\ln(2)}{-\ln(\overline{p})} + 1\right).
\end{align}
We have therefore obtained the desired lower bounds on $n$ for $\supp(\Tr_{\overline{C}}(\ket{\psi}\bra{\psi}))$ to contain a product state with probability greater than $\overline{p}$. We will shortly accumulate and simplify these lower bounds to obtain the theorem. 
Before doing this, we need to check the condition~\eqref{eq:binomfraccondition1}. Fully expanded and using the bounds~\eqref{eq:kc-1bound} on $k_{c-1}$, this condition is implied by
$$
(n-(c-1))\left(\frac{1}{2^{c}}-\frac{20e^5c(2^{c+1})}{\alpha^2\sqrt{2\pi\alpha(1-\alpha)(2^{c+1}-1)(n-c)}} \right)^{2^{c-1}/(n-2^c)} \geq \frac{n}{2^{c-1}}.
$$
Again, we assume the lower bound~\eqref{eq:nconditionnew2} on $n$, which implies~\eqref{eq:nconditionnew3}. Therefore the condition is implied by
$$
(n-c)\left( \frac{1}{2^{c+1}}\right)^{2^{c-1}/(n-2^c)} \geq \frac{n}{2^{c-1}},
$$
which is implied by 
\begin{align}\label{eq:ncondnew4}
    n \geq 2^{c+1}.
\end{align}
We can now collect the various lower bounds on $n$ we have assumed throughout this proof. These are~\eqref{eq:ncondnew1}, \eqref{eq:ncondnew21}, \eqref{eq:ncondnew22}, \eqref{eq:ncondnew23}, \eqref{eq:nconditionnew2}, \eqref{eq:ncondnew3} and \eqref{eq:ncondnew4}. We see that all the lower bounds except~\eqref{eq:ncondnew3} are implied by~\eqref{eq:nconditionnew2}. This leaves us with two lower bounds; inserting the definition~\eqref{eq:alphadefnew} of $\alpha$ in~\eqref{eq:nconditionnew2} and simplifying yields the lower bound in the theorem statement.
\end{proof}

\section{Proof of Theorem~\ref{thm:localunsat}}\label{sec:localunsat}

\subsection{Overview}

\begin{theorem*}[Restatement of Theorem~\ref{thm:localunsat}]
There exists a constant $c(k,\epsilon)$ (i.e. not depending on $n$) such that the following holds for large enough $n$. Let $([n],\{\Pi_s\})$ be any instance of quantum $k$-SAT which is $\epsilon$-far from satisfiable by a product state. Then, for a randomly chosen subset $C \in \sigma_{c(k,\epsilon)}([n])$, the instance is locally unsatisfiable by a product state at $C$ with probability $p > 0.75$.
\end{theorem*}
\noindent
We first remark that one idea for a proof of this theorem is to add the usual polynomial ground state energy bound to the definition of quantum $k$-SAT, then use the standard approach with a $\delta$-net to reduce the problem of satisfiability by a product state to an instance of classical $k'$-SAT. It is straightforwardly seen that this instance of $k'$-SAT is also $\epsilon$-far from satisfiable, so the result from~\cite{Alon2003} can be directly applied to show local unsatisfiability. The problem is that the constant $\delta$ determining the precision of the net, and therefore also the number $k'$, depends polynomially on $n$; this implies that the size of the subsets to be checked also depends at least polynomially on $n$. Since the time taken to check for a product state solution is exponential in the size of the subset, as we saw in the proof of Corollary~\ref{cor:testing}, we end up with time exponential in $n$. 

Instead, we give a proof which is heavily inspired by Alon and Shapira's proof of the classical analogue (to the point where we were able to lift most of the constants from their work), and which is in any case more general than the above approach since it does not require any lower bound on the ground state energy when an instance of quantum $k$-SAT is unsatisfiable. 

Firstly, we consider the problem of extending a local assignment, which in our setting is a product subspace. We define the notion of a bad qubit; that is, a qubit to which either the local assignment cannot be extended, or such that extension to that qubit greatly reduces the dimension of possible extensions to the other qubits. If, in the course of constructing a local assignment by extension, one extends through $5(2)^{k-1}/\epsilon$ bad qubits, then the assignment can no longer be extended (Lemma~\ref{lem:tree}). We show that if an instance of quantum $k$-SAT is $\epsilon$-far from satisfiable by a product state, then there are always more than $\epsilon n/5$ bad qubits with respect to any local assignment (Lemma~\ref{lem:enoughbad}). 

At this point, Alon and Shapira~\cite[Claim 2.4]{Alon2003} built a binary tree of possible assignments to a set of randomly picked variables and showed that each branch of the tree will run into more than $5(2)^{k-1}/\epsilon$ bad variables with high probability; they then used the union bound to conclude that there is no local assignment to those variables with high probability. This approach cannot be applied directly in the quantum setting because there are in general continuously many assignments to each variable (the space of possible assignments is a Hilbert space and not a binary set) and it is therefore not possible to build a finite tree. Instead, we use a backtracking approach. In order to eliminate all possible solutions in a finite process, we still need to discretise somehow; for this we use Lemma~\ref{lem:prodcount}, which implies that we need only eliminate a finite number of possible solutions.

\subsection{Definitions}

\begin{notation}
Throughout this section we write $|H|:= \dim(H)$ for the dimension of a Hilbert space $H$ and $|X|$ for the cardinality of a set $X$. When the context does not adequately distinguish these two meanings we explicitly use $\dim(H)$ for the dimension and $\#(X)$ for the cardinality. 

Let $V,W$ be two sets. We write $\sigma_k(\underline{V} \sqcup W)$ for the set of $k$-element subsets of $V \sqcup W$ containing all the elements of $V$; that is, underlining one of the sets in the union means all its elements must be contained in the $k$-element subsets. 

For a subset $X \subseteq [n]$ we write $H_X := (\mathbb{C}^2)^{\otimes |X|}$ for the Hilbert space of the qubits in $X$. 
\end{notation}

\begin{definition}
Fix a subset $S \subset [n]$, and let $H_S := (\mathbb{C}^{2})^{\otimes |S|}$ be the Hilbert space of the qubits in $S$. We say that an \emph{assignment} to $S$ is a product subspace $P_S \subseteq H_S$ such that every state in the subspace is a local solution for $([n],\{\Pi_s\})$ at $S$; that is, a subspace $P_S = V_1 \otimes \dots \otimes V_{|S|} \subseteq H_S$ such that every state $\ket{\phi_S} \in P_S$ satisfies $\sum_{s \in \sigma_{k}(S)} \Pi_s \ket{\phi_S} = 0$. 

We will often simplify language by calling a pair of a subset $S \subset [n]$ and an assignment $P_S \subset H_S$ a \emph{local assignment} $(S,P_S)$.
\end{definition}
\noindent
Following Alon and Shapira, we are now going to define some auxiliary Hilbert spaces that will help us to understand the possible extensions of a local assignment $(S,P_S)$. 
We will begin by choosing any set of qubits $\{x_1,\dots, x_j\}$ not in $S$, and defining a Hilbert space $L_{S,P_S}(x_1,\dots,x_j)$ of possible states of the qubits $x_1,\dots,x_j$ that are compatible with the assignment $(S,P_S)$.
\begin{definition}\label{def:LSPS}
For every $\{x_1,\dots,x_j\} \in \sigma_j([n] - S)$, where $1 \leq j \leq k-1$, we define $L_{S,P_S}(x_1,\dots,x_j) \subseteq H_{x_1,\dots,x_j}$ to be the subspace of all states $\ket{v_{x_1,\dots,x_j}} \in H_{x_1,\dots,x_j}$ (not necessarily product states!) such that
\begin{equation}\label{eq:bigl}
\sum_{s \in  \sigma_k(\underline{\{x_1,\dots,x_j\}} \sqcup  S) } \Pi_s (\ket{v_{x_1,\dots,x_j}} \otimes \ket{\phi_{S}}) = 0 \qquad \forall \ket{\phi_S} \in P_S.
\end{equation}
Note the underline in the sum: we are saying that $\ket{v_{x_1,\dots,x_j}} \otimes P_S$ is in the kernel of every projector on the qubits $\{x_1,\dots,x_j\} \sqcup S$ which acts on \emph{all} of the qubits $x_1,\dots,x_j$.
\end{definition}
\noindent
Now that we have our individual spaces $L_{S,P_S}(x_1,\dots,x_j)$ describing possible extensions to individual sets of qubits $\{x_1,\dots,x_j\}$, what we want to do is understand the space of all possible extensions to all subsets of $j$ qubits, so that we can understand how fast it shrinks as we extend an assignment. This motivates the following definition.
\begin{definition}\label{def:ljdef}
We write the formal orthogonal direct sum of Hilbert spaces using the following notation:
$$L_{S,P_S,j} := \bigoplus_{\{x_1,\dots,x_j\} \in \sigma_j([n]-S)} L_{S,P_S}(x_1,\dots,x_j).$$ 
\end{definition}
\noindent
The Hilbert space $L_{S,P_S,j}$ is the direct sum of all the spaces of possible $j$-qubit extensions of the assignment $(S,P_S)$. Its dimension measures `how many' possible $j$-qubit extensions there are to the local assignment $(S,P_S)$.

Now that we have defined our Hilbert spaces that characterise possible extensions to an assignment, we can consider what happens to them when we do extensions. The notion of extension is straightforward: 
\begin{itemize}
    \item We begin with some local assignment $(S,P_S)$.
    \item We choose any qubit $x \notin S$ such that $L_{S,P_S}(x) \neq \emptyset$. 
    \item We choose any subspace $V_x \subset L_{S,P_S}(x)$.
    \item We define the extended assignment $(\{x\} \sqcup S, V_x \otimes P_S)$, which (as is easy to check) is indeed a valid local assignment to the qubits $\{x\} \sqcup S$.
\end{itemize}
When we extend an assignment, the space of possible extensions will shrink. We want some number that will allow us to measure how much it shrinks. For this, we rely on a combinatorial ansatz of Alon and Shapira~\cite{Alon2003}.
\begin{definition}
Suppose we begin with a local assignment $(S,P_S)$ and extend it to a local assignment $(\{x\} \sqcup S, V_x \otimes P_S)$. We define the following numbers:
\begin{align*}
\delta_{S,P_S,x,V_x,j} &:= |L_{S,P_S,j}| - |L_{\{x\} \sqcup S,V_x \otimes P_S,j}|\\
&=\sum_{\{x_1,\dots,x_j\} \in \sigma_j([n]-(\{x\} \sqcup S))} |L_{S,P_S}(x_1,\dots,x_j)| - |L_{\{x\} \sqcup S, V_x \otimes P_S}(x_1,\dots,x_j)|\\
\delta_{S,P_S,x,V_x} &:= \sum_{j=1}^{k-1} \delta_{S,P_S,x,V_x,j}~ n^{k-j-1}\end{align*}
In words, $\delta_{S,P_S,x,V_x,j}$ measures how much the dimension of the space of possible $j$-qubit extensions has been reduced by extending the assignment $(S,P_S)$ to $(\{x\} \sqcup S, V_x \otimes P_S)$. The second number $\delta_{S,P_S,x,V_x}$, which does not depend on $j$, is a weighted sum over the $\delta_{S,P_S,x,V_x,j}$ which encapsulates the total reduction of possible extensions in a single number.
\end{definition}
\noindent
For Theorem~\ref{thm:localunsat}, we want to show that a product state solution does not exist. The basic idea of our argument is to try to gradually build a solution by extension and backtracking. We will show that we run out of possible extensions before we ever get to an assignment on all the $n$ qubits. 

In order to show this, we will (again following Alon and Shapira) define a notion of a `bad' qubit. By `bad', we mean that it is bad for creating an assignment by extension (and therefore good for our proof). There are two ways in which a qubit can be bad. The first is that we simply cannot extend to it, because there are no extensions that do not conflict with our current local assignment. We call such a qubit `conflicting'. The second is that we can extend to it, but any such extension will drastically reduce the number of possible extensions to the other qubits, as measured by the number $\delta_{S,P_S,x,V_X}$; we call such a qubit `heavy'. Formally, we have the following definition. 
\begin{definition} 
Let $(S,P_S)$ be a local assignment, and let 
$x \in ([n]-S)$ be any qubit. 
\begin{itemize}
\item We say that the qubit $x$ \emph{conflicts} with $(S,P_S)$ if $$L_{S,P_S}(x) = \emptyset.$$
\item If a qubit $x \in ([n]-S)$ does not conflict with $(S,P_S)$, we say that it is \emph{heavy} with respect to (w.r.t.) $(S,P_S)$ if $$\min_{V_x \subseteq L_{S,P_S}(x)}\delta_{S,P_S,x,V_x} > \epsilon n^{k-1}/5.$$
Here the $\epsilon$ is from the statement of Theorem~\ref{thm:localunsat}, which can be found at the beginning of this section.
\end{itemize}
We say that a qubit $x \in ([n]-S)$ is \emph{bad} w.r.t. $(S,P_S)$ if it is either heavy or conflicting.
\end{definition}
\noindent
Before commencing the proof, we will make one final notational definition, which introduces no new ideas. It will help simplify notation in the proof of Lemma~\ref{lem:enoughbad}, where we define a set of subsets $\Sigma \subset \sigma_k([n])$ and only consider projectors on subsets not in $\Sigma$.
\begin{definition}
Let $\Sigma \subset \sigma_k([n])$. 
Let $(S,P_S)$ be a local assignment, and let $\{x_1,\dots,x_j\} \in \sigma_{j}([n]-S)$. Recall the definition of $L_{S,P_S}(x_1,\dots,x_j)$ (Definition~\ref{def:LSPS}). Analogously, we define a subspace
$$L_{S,P_S}^{\Sigma}(x_1,\dots,x_j) \subseteq H_{x_1,\dots,x_j}$$ 
as the space of all states $\ket{v_{x_1,\dots,x_j}} \in H_{x_1,\dots,x_j}$ such that 
$$
\label{eq:biglsigma}
\sum_{s \in  \sigma_k(\underline{\{x_1,\dots,x_j\}} \sqcup  S) - \Sigma } \Pi_s (\ket{v_{x_1,\dots,x_j}} \otimes \ket{\phi_{S}}) = 0 \qquad \forall \ket{\phi_S} \in P_S.
$$
That is, the definition is identical to that of $L_{S,P_S}(x_1,\dots,x_j)$, but the projectors in $\Sigma$ are not considered. We therefore have 
$$
L_{S,P_S}(x_1,\dots,x_j) \subseteq L_{S,P_S}^{\Sigma}(x_1,\dots,x_j) \subseteq H_{x_1,\dots,x_j}.
$$
\end{definition}
\subsection{Part 1: Extending assignments}

Part 1 of the proof very closely follows Alon and Shapira. As discussed in the introduction, our aim in Part 1 is to show the following:
\begin{enumerate}
    \item If, in the course of constructing a local assignment by extension, one extends through $5(2)^{k-1}/\epsilon$ heavy qubits, then the assignment can no longer be extended (Lemma~\ref{lem:tree}). 
    \item If an instance of quantum $k$-SAT is $\epsilon$-far from satisfiable by a product state, then there are always more than $\epsilon n/5$ bad qubits with respect to any local assignment (Lemma~\ref{lem:enoughbad}). 
\end{enumerate}
In order to show the first result, we will define a `chain', which is precisely a local assignment that has been constructed by repeatedly extending through heavy qubits. 
\begin{definition}\label{def:chain}
A \emph{chain} of length $m$ is a local assignment $(\{y_1,\dots,y_{m}\},P_1 \otimes \dots \otimes P_{m})$ such that:
\begin{itemize}
    \item For each $2 \leq l \leq m$, each qubit $y_{l}$ is heavy with respect to the assignment $(\{y_1,\dots,y_{l-1}\},P_1 \otimes \dots \otimes P_{l-1})$.
    \item The qubit $y_1$ is heavy with respect to the empty assignment $(\emptyset,\emptyset)$.
\end{itemize}
\end{definition}
\noindent 
We can now prove the first lemma. Although it only applies to chains, i.e. assignments where every qubit is heavy with respect to the preceding assignment, that will be all we need for the proof. 
\begin{lemma}\label{lem:tree}
Let $(\{y_1,\dots,y_{\gamma}\},P_1 \otimes \dots \otimes P_{\gamma})$ be a chain of length $\gamma(k,\epsilon):= 5 (2)^{k-1}/\epsilon$. Then every qubit $x \in ([n]-\{y_1,\dots,y_{\gamma}\})$ is conflicting w.r.t. the chain.
\end{lemma}
\begin{proof}
Let us construct the chain $(\{y_1,\dots,y_{\gamma}\},P_1 \otimes \dots \otimes P_{\gamma})$ starting from the empty assignment $(\emptyset,\emptyset)$. Recalling Definition~\ref{def:ljdef}, we define:
\begin{align*}
W_0 &:= \sum_{j=1}^{k-1} |L_{\emptyset,\emptyset,j}| n^{k-j-1}\\
W_i&:= \sum_{j=1}^{k-1} |L_{\{y_1,\dots,y_i\},P_1\otimes \dots \otimes P_i,j}| \, n^{k-j-1}, ~~~~ 1 < i < \gamma(k, \epsilon) 
\end{align*}
The initial size of each $|L_{\emptyset,\emptyset,j}|$ is at most $\sum_{\{v_1,\dots,v_j\}\in \sigma_j([n])} 2^{j} = {n \choose j} 2^{j} \leq \frac{2^j}{j!} n^j$. Therefore 
$$
W_0 \leq \sum_{j=1}^{k-1} \frac{2^j}{j!} n^j n^{k-j-1} \leq 2^{k-1} n^{k-1}.
$$
Now, when we extend to the qubit $y_{i}$ from the assignment $(\{y_1,\dots,y_{i-1}\},P_1 \otimes \dots \otimes P_{i-1})$, we have, by definition of heaviness,
\begin{align*}
W_{i-1} - W_{i} &= \sum_{j=1}^{k-1} \left(|L_{\{y_1,\dots,y_{i-1}\},P_1\otimes \dots \otimes P_{i-1},j}| - |L_{\{y_1,\dots,y_i\},P_1\otimes \dots \otimes P_i,j}|\right)\, n^{k-j-1} \\
&= \sum_{j=1}^{k-1} \delta_{\{y_1,\dots,y_{i-1}\},P_1 \otimes \dots \otimes P_{i-1},\{y_i\},P_{i},j} ~n^{k-j-1} 
\\
&= \delta_{\{y_1,\dots,y_{i-1}\},P_1 \otimes \dots \otimes P_{i-1},\{y_i\},P_{i}} > \epsilon n^{k-1}/5.
\end{align*}
Since $W_i \leq W_{i-1} - \epsilon n^{k-1}/5$, by the time we have made $5 (2)^{k-1}/\epsilon$ extensions we must have $W_{\gamma}=0$, and so all qubits are conflicting. 
\end{proof}
\noindent
We can now show the second lemma. Again, this closely follows the analogous result in~\cite{Alon2003}.
\begin{lemma}\label{lem:enoughbad}
Let $([n], \{\Pi_s\})$ be an instance of quantum $k$-SAT which is $\epsilon$-far from satisfiable by a product state. Let $S \subset [n]$ and let $P_S$ be some local assignment to $S$. Then there are at least $\epsilon n/5$ bad qubits w.r.t. $(S,P_S)$.
\end{lemma}
\begin{proof}
Suppose that there are fewer than $\epsilon n/5$ bad qubits w.r.t $(S,P_S)$. We will define a subset $\Sigma \subset \sigma_k([n])$, where  $|\Sigma| < \epsilon n^{k}$, such that there exists a product state solution when we remove all the projectors in $\Sigma$, as in~\eqref{eq:efar}. This would contradict the assumption that the instance is $\epsilon$-far from satisfiable by a product state, and so the lemma will follow. 

Let $X_{nb} \subset ([n]-S)$ be the set of qubits which are not bad w.r.t $(S,P_S)$. 
The definition of $\Sigma$ is a two-step process. In the first step we will extend the assignment $(S,P_S)$ to the not bad qubits in $X_{nb}$ one at a time, while adding $k$-subsets to $\Sigma$ at each extension in a careful way. In the second step we will throw caution to the wind and simply add all the projectors containing the remaining bad qubits to $\Sigma$.

\paragraph{Step 1.} In this step we will extend the assignment $(S,P_S)$ to the qubits in $X_{nb}$ one at a time, while adding $k$-subsets to $\Sigma$ at each extension, according to the following prescription. 

Let $x \in X_{nb}$ be the first qubit to which we extend; we choose some $V_x \subseteq L_{S,P_S}(x)$ minimising $\delta_{S,P_S,x,V_x}$, and get the new assignment $(\{x\} \sqcup S, V_x \otimes P_S)$. Let us consider in turn each $\{x_1,\dots,x_j\} \in \sigma_j([n]-(\{x\} \sqcup S))$, where $1 \leq j \leq k-1$. Clearly \begin{align}\label{eq:linclusion}
L_{\{x\} \sqcup S,V_x \otimes P_S}(x_1,\dots,x_j) \subseteq L_{S, P_S}(x_1,\dots,x_j).
\end{align}
The prescription for adding subsets to $\Sigma$ is as follows.
\begin{itemize}
    \item If the inclusion~\eqref{eq:linclusion} is an equality we add no $k$-subsets to $\Sigma$. 
    \item If the inclusion~\eqref{eq:linclusion} is strict then we add all the subsets $s \in \sigma_{k}( \underline{\{x_1,\dots,x_j\}} \sqcup \underline{\{x\}} \sqcup S)$ to $\Sigma$. (Recall the underline notation means that these subsets must contain all the qubits $\{x_1,\dots x_j\}$, as well as the qubit $x$.)
\end{itemize}
We now have an equality $L_{\{x\} \sqcup S,V_x \otimes P_S}^{\Sigma}(x_1,\dots,x_j) = L_{S, P_S}(x_1,\dots,x_j)$, seen as follows. For any states $\ket{v_{x_1,\dots,x_j}} \in L_{S, P_S}(x_1,\dots,x_j)$, $\ket{v_x} \in V_x$ and $\ket{\phi_S} \in P_S$, we have:
\beas
&& \sum_{s \in \sigma_k(\underline{\{x_1,\dots,x_j\}} \sqcup \{x\} \sqcup S) - \Sigma} \Pi_s (\ket{v_{x_1,\dots,x_j}} \otimes \ket{v_x} \otimes \ket{\phi_{S}}) 
\\
&=& \sum_{s \in \sigma_k({\underline{\{x_1,\dots,x_j\}} \sqcup S}) } \Pi_s (\ket{v_{x_1,\dots,x_j}} \otimes \ket{v_x} \otimes \ket{\phi_{S}})
\\
&& \, + \sum_{s \in \sigma_k(\underline{\{x_1,\dots,x_j\}} \sqcup \underline{\{x\}} \sqcup S)} \Pi_s (\ket{v_{x_1,\dots,x_j}} \otimes \ket{v_x} \otimes \ket{\phi_{S}})
\\
&& \, -\sum_{s \in \sigma_k(\underline{\{x_1,\dots,x_j\}} \sqcup \underline{\{x\}} \sqcup S))} \Pi_s (\ket{v_{x_1,\dots,x_j}} \otimes \ket{v_x} \otimes \ket{\phi_{S}})
\\
&=& 0.
\eeas
For the second equality, the second and third terms cancel and the first term is zero by definition of $L_{S,P_S}(x_1,\dots,x_j)$. 

We do this for all $\{x_1,\dots,x_j\} \in \sigma_j([n]-(\{x\} \sqcup S))$, where $1 \leq j \leq k-1$, and in the end we have a new assignment $(\{x\} \sqcup S, V_x \otimes P_S)$ and an updated set $\Sigma$ of removed projectors such that 
\begin{align}\label{eq:projremovedeq}
L_{\{x\} \sqcup S,V_x \otimes P_S}^{\Sigma}(x_1,\dots,x_j) = L_{S, P_S}(x_1,\dots,x_j), ~~~~\forall\{x_1,\dots,x_j\} \in \sigma_j([n]-(\{x\} \sqcup S)).
\end{align}
We now claim that each qubit in $X_{nb} - \{x\}$ remains not bad for the assignment $(\{x\} \sqcup S, V_x \otimes P_s)$, when we remove the projectors in $\Sigma$. Indeed, from~\eqref{eq:projremovedeq} it is clear that none of these qubits will ever conflict. It remains to show that none of them will become heavy.

To see this, suppose that we extend the assignment to
$(\{y_1 ,\dots y_m\} \sqcup S, V_{y_1} \otimes \dots \otimes V_{y_m} \otimes P_S)$,
removing the relevant subsets as above, where $\{y_1, \dots y_m\} \subset X_{nb}$, and then extend to $x \in (X_{nb} - (\{y_1 ,\dots y_m\} \sqcup S))$ with subspace $V_x \subseteq L_{\{y_1 ,\dots y_m\} \sqcup S, V_{y_1} \otimes \dots \otimes V_{y_m} \otimes P_S}^{\Sigma}(x) = L_{S,P_S}(x)$ (the equality is due to the addition of subsets to $\Sigma$ detailed above). Alternatively, we could have extended to $x$ directly from $(S,P_S)$, without extending through $\{y_1,\dots,y_m\}$ first. To show that the qubit $x$ will not be heavy after extending through $\{y_1,\dots,y_m\}$ it is sufficient to show, for any $\{x_1,\dots,x_j\} \in [n]-(\{x,y_1,\dots,y_m\} \sqcup S)$, that there is an inclusion:
\begin{equation}
\label{eq:heavyinclusion}
L_{\{x\} \sqcup S, V_x \otimes P_S}(x_1,\dots,x_j) \subseteq L_{\{x,y_1 ,\dots y_m\} \sqcup S, V_x \otimes V_{y_1} \otimes \dots \otimes V_{y_m} \otimes P_S}^{\Sigma}(x_1,\dots,x_j).
\end{equation}
Here the subspaces are taken prior to the removal of subsets following the extension to qubit $x$; but on the RHS we have removed all the relevant subsets following the previous extension to $\{y_1,\dots,y_m\}$.

We will now prove the inclusion. Let $\ket{v_{x_1,\dots,x_j}} \in L_{\{x\} \sqcup S, V_x \otimes P_S}(x_1,\dots,x_j)$. We will consider two cases. The first is where $j=k-1$. In the following equations $\Sigma$ is defined after extension to $\{y_1,\dots,y_m\}$. For all $\ket{v_{x}} \in V_x$, $\ket{v_{y_i}} \in V_{y_i}$ and $\ket{\phi_S} \in P_S$, we have:
\beas
&& \sum_{s \in \underline{\{x_1,\dots,x_{k-1}\}} \sqcup \{x\} \sqcup \{y_1,\dots,y_m\} \sqcup S  - \Sigma} \Pi_s (\ket{v_{x_1,\dots,x_{k-1}}} \otimes \ket{v_x} \otimes \ket{v_{y_1}} \otimes \dots \otimes \ket{v_{y_m}}  \otimes \ket{\phi_S})
\\
&=& 
\sum_{s \in \underline{\{x_1,\dots,x_{k-1}\}}  \sqcup \{y_1,\dots,y_m\} \sqcup S - \Sigma} \Pi_s (\ket{v_{x_1,\dots,x_{k-1}}} \otimes \ket{v_x} \otimes \ket{v_{y_1}} \otimes \dots \otimes \ket{v_{y_m}}  \otimes \ket{\phi_S})
\\
&& \, +
\Pi_{\{x_1,\dots,x_{k-1},x\}} (\ket{v_{x_1,\dots,x_{k-1}}} \otimes \ket{v_x} \otimes \ket{v_{y_1}} \otimes \dots \otimes \ket{v_{y_m}}  \otimes \ket{\phi_S})
\\
&=& 0 .
\eeas
Here the first term on the second line is zero since $\Sigma$ is defined so that $$L_{\{y_1 ,\dots y_m\} \sqcup S, V_{y_1} \otimes \dots \otimes V_{y_m} \otimes P_S}^{\Sigma}(x_1,\dots,x_{k-1}) = L_{S,P_S}(x_1,\dots,x_{k-1});$$ the second term is zero by definition of $L_{\{x\} \sqcup S, V_x \otimes P_S}(x_1,\dots,x_{k-1})$. Therefore $\ket{v_{x_1,\dots,x_j}} \in L_{\{x,y_1 ,\dots y_m\} \sqcup S, V_x \otimes V_{y_1} \otimes \dots \otimes V_{y_m} \otimes P_S}(x_1,\dots,x_j)$ and we obtain the inclusion~\eqref{eq:heavyinclusion}.

On the other hand, suppose that $j < k-1$. 
Then for all $\ket{v_{x}} \in V_x$, $\ket{v_{y_i}} \in V_{y_i}$ and $\ket{\phi_S} \in P_S$ we have the following equation:
\beas
&& \sum_{s \in \underline{\{x_1,\dots,x_{j}\}} \sqcup \{x\} \sqcup \{y_1,\dots,y_m\} \sqcup S  - \Sigma} \Pi_s (\ket{v_{x_1,\dots,x_{j}}} \otimes \ket{v_x} \otimes \ket{v_{y_1}} \otimes \dots \otimes \ket{v_{y_m}}  \otimes \ket{\phi_S})
\\
&=&
\sum_{s \in \underline{\{x_1,\dots,x_{j}\}} \sqcup \{y_1,\dots,y_m\} \sqcup S  - \Sigma} \Pi_s (\ket{v_{x_1,\dots,x_{j}}} \otimes \ket{v_x} \otimes \ket{v_{y_1}} \otimes \dots \otimes \ket{v_{y_m}}  \otimes \ket{\phi_S})
\\
&& \, +
\sum_{s \in \underline{\{x_1,\dots,x_{j},x\}} \sqcup \{y_1,\dots,y_m\} \sqcup S  - \Sigma} \Pi_s (\ket{v_{x_1,\dots,x_{j}}} \otimes \ket{v_x} \otimes \ket{v_{y_1}} \otimes \dots \otimes \ket{v_{y_m}}  \otimes \ket{\phi_S})
\\
&=& 0.
\eeas
Here the terms in the second line are both zero since $\Sigma$ is defined so that:
\begin{align*}
L_{\{y_1 ,\dots y_m\} \sqcup S, V_{y_1} \otimes \dots \otimes V_{y_m} \otimes P_S}^{\Sigma}(x_1,\dots,x_{j}) &= L_{S,P_S}(x_1,\dots,x_{j}),
\\
L_{\{y_1 ,\dots y_m\} \sqcup S, V_{y_1} \otimes \dots \otimes V_{y_m} \otimes P_S}^{\Sigma}(x_1,\dots,x_{j},x) &= L_{S,P_S}(x_1,\dots,x_{j},x).
\end{align*}
Therefore $\ket{v_{x_1,\dots,x_j}} \in L_{\{x,y_1 ,\dots y_m\} \sqcup S, V_x \otimes V_{y_1} \otimes \dots \otimes V_{y_m} \otimes P_S}^{\Sigma}(x_1,\dots,x_j)$ and we obtain the inclusion~\eqref{eq:heavyinclusion}. We have shown that a qubit in $X_{nb}$ cannot become bad upon extension.

We therefore continue extending until we have an assignment on all the qubits in $X_{nb} \sqcup S$. At this point we have accumulated a corresponding set of projections $\Sigma$, according to the prescription given. 

\paragraph{Step 2.} In the second step we simply add to $\Sigma$ all projections containing any qubit in $[n] - (X_{nb} \sqcup S)$. We then extend the assignment to the whole of $[n]$, which is trivial since there are no longer any projections involving the remaining qubits. We thereby obtain a product state $\ket{\phi} \in (\mathbb{C}^2)^{\otimes n}$ satisfying~\eqref{eq:efar}.

\paragraph{Analysis.} We must show that in the course of this process we added less than $\epsilon n^k$ subsets to $\Sigma$. In the first step, after extending to a qubit $x$ we added all the subsets $s \in \sigma_k( \underline{\{x_1,\dots,x_j\}} \sqcup \underline{x} \sqcup S)$ each time that $|L_{\{x\} \sqcup S,V_x \otimes P_S}(x_1,\dots,x_j)|$ was smaller than $|L_{S, P_S}(x_1,\dots,x_j)|$. This is at most ${n \choose k-j-1} \leq n^{k-j-1}$ subsets each time, since $|S| \leq n$. Since the qubit $x$ was never heavy, we therefore added at most $\sum_{j=1}^{k-1} \delta_{S,P_S,x,V_x,j} n^{k-j-1}= \delta_{S,P_S,x,V_x} \leq \epsilon n^{k-1}/5$ projectors every time we extended. As there are at most $n$ qubits in $X_{nb}$, in the first step we added less than $\epsilon n^k/5$ subsets to $\Sigma$. 

In the second step, we added all the subsets which contained a qubit not in $S \sqcup V_{nb}$. Since by assumption there were $\leq \epsilon n/5$ bad qubits to begin with, we therefore added at most $(\epsilon n/5) {n \choose k-1} \leq \epsilon n^k/5$ subsets in the second step.  Altogether we added less than $\epsilon n^k$ subsets, as claimed. 
\end{proof}

\subsection{Part 2: Backtracking}

This is the point at which our proof diverges from that given in~\cite{Alon2003}, for the reasons given at the start of this section. We are going to use a different argument --- by gradually attempting to build a product state solution by extension and backtracking, we will eliminate all possible solutions. 

Since the space of possible extensions is continuous, we need a way to discretise. The following lemma solves this problem. 
For a state $\ket{\phi} \in \mathbb{C}^2$, we write $\left\langle\ket{\phi}\right\rangle \in \mathbb{P}^1$ for the corresponding point on the complex projective line (a.k.a. the Bloch sphere). 
\begin{lemma}\label{lem:prodcount}
Let $L \subseteq (\mathbb{C}^2)^{\otimes m}$ be a subspace. We define
$$
X := \{x \in \mathbb{P}^1 : ~\exists \{\ket{\phi_i} \in \mathbb{C}^2\}_{i=1}^m ~\textrm{s.t.}~ \ket{\phi_1} \otimes \ket{\phi_2} \otimes \dots \otimes \ket{\phi_m} \in L ~ \textrm{and} ~\left\langle{\ket{\phi_1}}\right\rangle = x\}.
$$
Then either (i) $X = \mathbb{P}^1$, or (ii) $\#(X) \leq m!$. The bound in case (ii) is tight.
\end{lemma}
\begin{proof}
It is straightforward to show that $X$ is either finite or covers the whole space. Indeed, the subspace $L$ corresponds to a projective linear subspace $\mathbb{P}(L) \subseteq \mathbb{P}^{2^m-1}$. By the Segre embedding the product states form a projective subvariety $\Sigma_m := \mathbb{P}^1 \times \dots \times \mathbb{P}^1 \subset \mathbb{P}^{2m-1}$. But it is an elementary result in algebraic geometry that the projection map $\pi_1: \mathbb{P}^1 \times \dots \times \mathbb{P}^1 \to \mathbb{P}^1$ is closed in the Zariski topology, since every projective variety is complete \cite[\S{}2 Thm. 4.9]{Hartshorne2013}. The intersection $\mathbb{P}(L) \cap \Sigma_m$ is an algebraic set in $\Sigma_m$, and therefore by closure of the projection map, the image under the projection will be an algebraic set in $\mathbb{P}^1$. This set will either be zero-dimensional, in which case it is a finite set of points; or one-dimensional, in which case it must be the whole of $\mathbb{P}^1$. 

The only thing left is to bound the number of points in the case where the image is zero-dimensional. We first note that for any variety $S \subset \Sigma_m$, the image $\pi_1(S)$ is irreducible; indeed, were it not irreducible then by continuity of $\pi_1$ the fibres would give a nontrivial decomposition of $S$. If $\pi_1(S) \neq \mathbb{P}^1$, it must therefore be a single point. It follows that all we need to do is bound the number of irreducible components of $\mathbb{P}(L) \cap \Sigma_m$.  We know that $\mathbb{P}(L)$ is an intersection of hyperplanes $\{P_i\}_{i=1}^{2^{m}-|L|}$. We use~\cite[\S{}1 Thm. 7.7]{Hartshorne2013}, which yields the following statement in our special case. For any subvariety $V \subset \mathbb{P}^{2^m-1}$ of dimension $\geq 1$, and for any hyperplane $P \subset \mathbb{P}^{2^m-1}$ not containing $V$, let $Z_1,\dots,Z_s$ be the irreducible components of $V \cap P$; then:
\begin{align}\label{eq:intersectiondeg}
    \sum_{j=1}^s \textrm{deg}(Z_j) \leq \textrm{deg}(V).
\end{align}
Here $\textrm{deg}(Z_j), \textrm{deg}(V) \in \mathbb{N}_{>0}$ are the degrees of the varieties in $\mathbb{P}^{2^m-1}$. If $P$ contains $V$ then the intersection is just $V$ again. We observe that the irreducible components of $\mathbb{P}(L) \cap \Sigma_m$ can be obtained by the following process: take the intersection $\Sigma_m \cap P_1$; then for each of the irreducible components of that intersection take the intersection with $P_2$; then for each of the irreducible components of that intersection take the intersection with $P_3$; etc. Since $\deg(\Sigma_m) = m!$, by~\eqref{eq:intersectiondeg} we finish with at most $m!$ irreducible components (which would all have degree 1). The bound follows. 

To see that the bound is tight, observe that by definition of the degree, the variety $\Sigma_m$ intersects a generic linear subvariety $\mathbb{P}(V) \subset \mathbb{P}^{2^m-1}$ of dimension $2^m-(m+1)$ in precisely $m!$ points, and generically these points will be mapped to $m!$ different points in $\mathbb{P}^1$ by the projection $\pi_1$.
\end{proof}
\noindent
We can now, at last, prove the theorem. We first make a notational definition. 
\begin{definition}
Let $(S,P_S)$ be some local assignment. For any $X \subseteq ([n]-S)$ we define $\overline{L}_{S,P_S}(X) \subseteq H_{X}$ to be the space of all states $\ket{v_{X}} \in H_{X}$ such that
\begin{equation}
\sum_{s \in  \sigma_k(X \sqcup  S) } \Pi_s (\ket{v_{X}} \otimes \ket{\phi_{S}}) = 0 \qquad \forall \ket{\phi_S} \in P_S.
\end{equation}
(The difference between this and~\eqref{eq:bigl} is that i) the set $X$ can be of any cardinality, and ii) all the $k$-subsets of $X \sqcup S$ are included, not just those containing every element of $X$.)
\end{definition}

\begin{theorem*}[Restatement of Theorem~\ref{thm:localunsat}]
There exists a constant $c(k,\epsilon)$ (i.e. not depending on $n$) such that the following holds for large enough $n$. Let $([n],\{\Pi_s\})$ be any instance of quantum $k$-SAT which is $\epsilon$-far from satisfiable by a product state. Then, for a randomly chosen subset $C \in \sigma_{c(k,\epsilon)}([n])$, the instance is locally unsatisfiable by a product state at $C$ with probability $p > 0.75$.
\end{theorem*}
\begin{proof}[Proof of Theorem~\ref{thm:localunsat}]
The idea of the proof is as follows. We will pick qubits at random from $[n]$, one by one. For any local assignment on the qubits which have already been picked, we know by Lemma~\ref{lem:enoughbad} that the next qubit we pick has an $\epsilon/5$ chance of being bad w.r.t. that assignment. By varying the local assignment and picking bad qubits, we will show that, after enough qubits have been picked, with high probability there can be no local assignment to all those qubits. This is essentially a backtracking argument where we eliminate assignments using Lemma~\ref{lem:tree}, which shows that we cannot extend a chain through more than $\gamma(k,\epsilon)$ heavy variables. We have not endeavoured to optimise this procedure and it can certainly be improved, although whether one can bring the constant $c(k,\epsilon)$ down to the level in the classical setting~\cite[Thm. 1.1]{Alon2003} using such a backtracking approach is by no means clear. 

To formalise this, recall from Definition~\ref{def:chain} that a chain of length $m$ is a local assignment $(\{y_1,\dots,y_{m}\},P_1 \otimes \dots \otimes P_{m})$ such that each qubit is heavy with respect to the assignment on the qubits before it. We will now define an \emph{elimination procedure} for chains of length $m$. 

\begin{definition}
    An \emph{elimination procedure for chains of length $m$} will, for any chain $(\{y_1,\dots,y_{m}\},P_1 \otimes \dots \otimes P_{m})$, produce a set of qubits $\Gamma \in ([n] - \{y_1, \dots, y_m\})$ with a guarantee that there is no local assignment to the qubits $\{y_1, \dots, y_m\} \sqcup \Gamma$ which assigns $P_1 \otimes \cdots \otimes P_{m-1}$ to the qubits $\{y_1,\dots,y_{m-1}\}$.

    We say that an elimination procedure procedure for chains of length $m$ has  constant $\phi(m) \in \mathbb{N}$ if, for any chain of length $m$, $|\Gamma| \leq \phi(m)$.
\end{definition}
\noindent
The elimination procedures we will consider here are all defined by a simple rule. At the beginning of the procedure we set $\Gamma = \emptyset$. At each step of the procedure, we do the following:
\begin{itemize}
    \item Inspect the current $\Gamma$, and choose a certain local assignment $(S,P_S)$ on a subset $S \subset \{y_1,\dots,y_m\} \sqcup \Gamma$, according to the rule.
    \item Pick at random a qubit $x \in ([n] - (\{y_1,\dots,y_m\} \sqcup \Gamma))$ which is bad w.r.t. the chosen assignment $(S,P_S)$.
    \item Add the qubit $x$ to $\Gamma$.
\end{itemize}
When $\phi(m)$ is the constant associated to the elimination procedure, we can stop after $\phi(m)$ steps and conclude that there is no local assignment to $\{y_1,\dots,y_{m}\} \sqcup \Gamma$ which assigns $P_1 \otimes \dots \otimes P_{m-1}$ to the qubits $\{y_1,\dots,y_{m-1}\}$.
\begin{lemma}    
An elimination procedure exists for chains of length $\gamma:=\gamma(k,\epsilon)$, with constant $\phi(\gamma) = 1$.
\end{lemma}
\begin{proof}
For any such chain $(\{y_1,\dots,y_{\gamma}\},P_1 \otimes \dots \otimes P_{\gamma})$, as soon as we add any qubit at all to $\Gamma$ we can conclude, by Lemma~\ref{lem:tree}, that there is no local assignment to $\{y_1,\dots,y_{\gamma}\} \sqcup \Gamma$ which assigns $P_1 \otimes \dots \otimes P_{\gamma-1}$ to the qubits $\{y_1,\dots,y_{\gamma-1}\}$. 
\end{proof}
\begin{lemma}
    If an elimination procedure exists for chains of length $m$ with constant $\phi(m)$, then there exists an elimination procedure for chains of length $m-1$ with constant $\phi(m-1) := ((\phi(m)+1)!+1)(\phi(m)+1)$.
\end{lemma}
\begin{proof}
The procedure is defined as follows. Let $(\{y_1,\dots,y_{m-1}\},P_1 \otimes \dots \otimes P_{m-1})$ be any chain of length $m-1$, and let $\ket{p_{m-1}^0} \in P_{m-1}$ be any state.
\begin{enumerate}
    \item We pick at random a bad qubit $y_m^0 \in ([n] - \{y_1,\dots,y_{m-1}\})$ with respect to the local assignment $(\{y_1,\dots,y_{m-1}\},P_1 \otimes \dots \otimes P_{m-2} \otimes \ket{p^0_{m-1}})$. If it is conflicting, then set $\Gamma_0 := \{y_m^0\}$ and skip to 3. If it is heavy, continue to 2.
    \item If the qubit $y_m^0$ is heavy, then we perform the length $m$ elimination procedure on the chain $$(\{y_1,\dots,y_{m-1},y_m^0\},P_1 \otimes \dots \otimes  P_{m-2}  \otimes \ket{p^0_{m-1}} \otimes \overline{L}_{\{y_1,\dots,y_{m-1}\},P_1 \otimes \dots \otimes P_{m-2} \otimes \ket{p^0_{m-1}}}(\{y_m^0\})).$$ 
    This gives us a set $\Gamma_m^0 \subseteq ([n] - \{y_1,\dots,y_m^0\})$ of qubits, where $|\Gamma_m^0| \leq \phi(m)$, such that there is no local assignment to $\{y_1,\dots,y_m^0\}\sqcup \Gamma_m^0$ which assigns $P_1 \otimes \dots \otimes P_{m-2} \otimes \ket{p^0_{m-1}}$ to the qubits $\{y_1,\dots,y_{m-1}\}$. Set $\Gamma_0 := \{y_m^0\} \sqcup \Gamma_m^0$ and continue to 3.
    \item In the last two steps, we have produced a set of qubits $\Gamma_0$ and concluded that there is no local assignment to the qubits $\{y_1,\dots,y_{m-1}\} \sqcup \Gamma_0$ which assigns $P_1 \otimes \dots \otimes P_{m-2} \otimes \ket{p^0_{m-1}}$ to the qubits $\{y_1,\dots,y_{m-1}\}$.

    Now consider the subspace
$$
L:= \overline{L}_{\{y_1,\dots,y_{m-2}\},P_1 \otimes \dots \otimes P_{m-2}}(\{y_{m-1}\}\sqcup \Gamma_0) \subseteq H_{y_{m-1}} \otimes H_{\Gamma_0} \cong (\mathbb{C}^{2})^{\otimes (|\Gamma_0|+1)}.
$$
Let $$X = \{\left\langle\ket{\phi_0}\right\rangle \in \mathbb{P}^1~:~\ket{\phi_0}\otimes \ket{\phi_1} \otimes \dots \otimes \ket{\phi_{|\Gamma_0|}} \in L\}.$$ 
Since, as we just showed, there is no local assignment to $\{y_1,\dots,y_{m-1}\}\sqcup \Gamma_0$ which assigns $P_1 \otimes \dots \otimes P_{m-2} \otimes \ket{p^0_{m-1}}$ to the qubits $\{y_1,\dots,y_{m-1}\}$, we know that $X \neq \mathbb{P}^1$, since $\left\langle \ket{p^0_{m-1}} \right\rangle \notin X$. But then, by Lemma~\ref{lem:prodcount}, we have $\#(X) \leq (|\Gamma_0|+1)!$. 

\item We will now eliminate all these possible assignments to the qubit $y_{m-1}$, one by one. Obviously we have to do more elimination if there are more possible assignments, so the worst case is where $\#(X) \leq (|\Gamma_0|+1)!$. In this case, let \begin{align}\label{eq:possassignments}\{\ket{p_{m-1}^i} \in H_{y_{m-1}}\}_{i=1}^{(|\Gamma_0|+1)!}\end{align}
be any state vectors corresponding to those points on the Bloch sphere. 

The approach we use is identical to that in Steps 1 and 2, where we eliminated the possible assignment $\ket{p^0_{m-1}}$. For each $1 \leq i \leq (|\Gamma_0|+1)!$, in ascending order, we perform the following process:
\begin{enumerate}
    \item Pick at random a qubit $y_{m}^i \in ([n]-(\sqcup_{j=0}^{i-1}\Gamma_j \sqcup \{y_1,\dots,y_{m-1}\})$ which is bad with respect to the local assignment $(\{y_1,\dots,y_{m-1}\},P_1 \otimes \dots \otimes P_{m-2} \otimes \ket{p_{m-1}^i})$. If it is conflicting, then set $\Gamma_i := \{y_m^i\}$ and move to the $(i+1)$-th possible assignment~\eqref{eq:possassignments}. If it is heavy, continue to the next step.
    \item If the qubit $y_m^i$ is heavy, perform the length $m$ elimination procedure on the chain $$(\{y_1,\dots,y_{m-1},y_m^i\},P_1 \otimes \dots \otimes  P_{m-2}  \otimes \ket{p^i_{m-1}} \otimes \overline{L}_{\{y_1,\dots,y_{m-1}\},P_1 \otimes \dots \otimes P_{m-2} \otimes \ket{p^i_{m-1}}}(\{y_m^i\})).$$ 
    This gives us a set $\Gamma_m^i \subseteq ([n]-(\sqcup_{j=0}^{i-1}\Gamma_j \sqcup \{y_1,\dots,y_m^i\})$ of qubits, where $|\Gamma_m^i| \leq \phi(m)$, such that there is no local assignment to $\{y_1,\dots,y_m^i\}\sqcup \Gamma_m^i$ which assigns $P_1 \otimes \dots \otimes P_{m-2} \otimes \ket{p^i_{m-1}}$ to the qubits $\{y_1,\dots,y_{m-1}\}$. Set $\Gamma_i := \{y_m^i\} \sqcup \Gamma_m^i$ and move to the $(i+1)$-th possible assignment~\eqref{eq:possassignments}.
\end{enumerate}
After we have run through all the possible assignments, we set $\Gamma := \sqcup_{i = 0}^{(|\Gamma_0| + 1)!} \Gamma_i$.
\end{enumerate}
The procedure is now finished, and we conclude that there is no local assignment to $\{y_1,\dots,y_{m-1}\} \sqcup \Gamma$ which assigns $P_1 \otimes \dots \otimes P_{m-2}$ to the qubits $\{y_1,\dots,y_{m-2}\}$. We now check the constant for this procedure. Since $\Gamma = \bigsqcup_{i=0}^{(|\Gamma_0|+1)!}\Gamma_i$, and we have $|\Gamma_i| \leq (\phi(m) +1)$ for all $i$, we conclude that 
$$|\Gamma| \leq ((\phi(m)+1)!+1)(\phi(m)+1) =: \phi(m-1).$$
\end{proof} 
\begin{corollary} 
There exists an elimination procedure for every $1 \leq m \leq \gamma$, with constant $\phi(m)$ defined by the recurrence relation $$\phi(m-1) = ((\phi(m)+1)!+1)(\phi(m)+1).$$
\end{corollary}
\noindent
Now that we have these elimination procedures, we can prove our theorem. Start with the empty assignment $(\emptyset,\emptyset)$. Pick a qubit $x \in [n]$ which is bad w.r.t. $(\emptyset,\emptyset)$. If it conflicts, then there is no local assignment to  $\{x\}$ and we are done. If it is heavy,  use the length 1 elimination procedure for the chain $(\{x\}, \overline{L}_{\emptyset,\emptyset}(\{x\}))$ to obtain a set $\Gamma \subset ([n] - \{x\})$, where $|\Gamma| \leq \phi(1)$, such that there is no local assignment on $\{x\} \sqcup \Gamma$.

We will now show that if we pick qubits at random, one at a time, then they will implement this `elimination procedure for the empty assignment' with high probability. Indeed, we recall from Lemma~\ref{lem:enoughbad} that, with respect to any local assignment $(S,P_S)$, there are more than $\epsilon n/5$ bad qubits in $([n]-S)$. For large enough $n$, every time we pick a new qubit there is therefore a roughly $\epsilon/5$ chance of it being bad w.r.t. any local assignment. The elimination procedure for the empty assignment requires us to pick a bad qubit w.r.t. some local assignment $\phi(0) := \phi(1)+1$ times. Using the Chernoff bound for the binomial distribution, we find that the probability that we have not picked $\phi(0)$ bad qubits after $c(k,\epsilon):=5\phi(0)/\epsilon+1$ qubits have been picked is less than $0.25$.
\end{proof}

\bibliographystyle{alphaurl}
\bibliography{bibliography}

\newcommand{\etalchar}[1]{$^{#1}$}
\begin{thebibliography}{CCD{\etalchar{+}}11}

\bibitem[AAV13]{aharonov13}
Dorit Aharonov, Itai Arad, and Thomas Vidick.
\newblock Guest column: the quantum {P}{C}{P} conjecture.
\newblock {\em SIGACT News}, 44(2):47–79, June 2013.
\newblock \href {https://doi.org/10.1145/2491533.2491549}
  {\path{doi:10.1145/2491533.2491549}}.

\bibitem[AB22]{anshu22}
Anurag Anshu and Nikolas~P. Breuckmann.
\newblock {A construction of combinatorial NLTS}.
\newblock {\em Journal of Mathematical Physics}, 63(12):122201, 12 2022.
\newblock \href {https://doi.org/10.1063/5.0113731}
  {\path{doi:10.1063/5.0113731}}.

\bibitem[ABN23]{anshu23}
Anurag Anshu, Nikolas~P. Breuckmann, and Chinmay Nirkhe.
\newblock {NLTS Hamiltonians from Good Quantum Codes}.
\newblock In {\em Proceedings of the 55th Annual ACM Symposium on Theory of
  Computing}, STOC 2023, page 1090–1096, New York, NY, USA, 2023. Association
  for Computing Machinery.
\newblock \href {https://doi.org/10.1145/3564246.3585114}
  {\path{doi:10.1145/3564246.3585114}}.

\bibitem[AdGS21]{Aldi2021}
Marco {Aldi}, Niel {de Beaudrap}, Sevag {Gharibian}, and Seyran {Saeedi}.
\newblock On efficiently solvable cases of quantum k-{S}{A}{T}.
\newblock {\em Communications in Mathematical Physics}, 381(1):209--256,
  January 2021.
\newblock \href {http://arxiv.org/abs/1712.09617} {\path{arXiv:1712.09617}},
  \href {https://doi.org/10.1007/s00220-020-03843-9}
  {\path{doi:10.1007/s00220-020-03843-9}}.

\bibitem[AKS12]{Ambainis2012}
Andris Ambainis, Julia Kempe, and Or~Sattath.
\newblock A quantum {L}ov{\'a}sz local lemma.
\newblock {\em Journal of the ACM}, 59(5):1--24, 2012.
\newblock \href {http://arxiv.org/abs/0911.1696} {\path{arXiv:0911.1696}},
  \href {https://doi.org/10.1145/2371656.2371659}
  {\path{doi:10.1145/2371656.2371659}}.

\bibitem[AS03]{Alon2003}
Noga Alon and Asaf Shapira.
\newblock Testing satisfiability.
\newblock {\em Journal of Algorithms}, 47(2):87--103, 2003.
\newblock URL: \url{https://www.tau.ac.il/~nogaa/PDFS/asafsodaproc2.pdf}, \href
  {https://doi.org/10.1016/S0196-6774(03)00019-1}
  {\path{doi:10.1016/S0196-6774(03)00019-1}}.

\bibitem[ASSZ18]{Arad2018}
Itai Arad, Miklos Santha, Aarthi Sundaram, and Shengyu Zhang.
\newblock Linear-time algorithm for quantum 2{S}{A}{T}.
\newblock {\em Theory of Computing}, 14(1):1--27, 2018.
\newblock \href {http://arxiv.org/abs/1508.06340} {\path{arXiv:1508.06340}},
  \href {https://doi.org/10.4086/toc.2018.v014a001}
  {\path{doi:10.4086/toc.2018.v014a001}}.

\bibitem[ATL11]{Augusiak2011}
R~Augusiak, J~Tura, and M~Lewenstein.
\newblock A note on the optimality of decomposable entanglement witnesses and
  completely entangled subspaces.
\newblock {\em Journal of Physics A: Mathematical and Theoretical},
  44(21):212001, April 2011.
\newblock \href {http://arxiv.org/abs/1012.3786} {\path{arXiv:1012.3786}},
  \href {https://doi.org/10.1088/1751-8113/44/21/212001}
  {\path{doi:10.1088/1751-8113/44/21/212001}}.

\bibitem[BCY11]{Brandao2011}
Fernando~G.S.L. Brandao, Matthias Christandl, and Jon Yard.
\newblock Faithful squashed entanglement.
\newblock {\em Communications in Mathematical Physics}, 306(3):805--830, 2011.
\newblock \href {http://arxiv.org/abs/1010.1750} {\path{arXiv:1010.1750}},
  \href {https://doi.org/10.1007/s00220-011-1302-1}
  {\path{doi:10.1007/s00220-011-1302-1}}.

\bibitem[BE15]{Bollobas2015}
B\'{e}la Bollob\'{a}s and Tom Eccles.
\newblock Partial shadows of set systems.
\newblock {\em Combinatorics, Probability and Computing}, 24(5):825–828,
  2015.
\newblock \href {https://doi.org/10.1017/S0963548314000790}
  {\path{doi:10.1017/S0963548314000790}}.

\bibitem[BFS15]{BARDET201549}
Magali Bardet, Jean-Charles Faugere, and Bruno Salvy.
\newblock {On the complexity of the F5 Gr\"obner basis algorithm}.
\newblock {\em Journal of Symbolic Computation}, 70:49--70, 2015.
\newblock \href {http://arxiv.org/abs/1312.1655} {\path{arXiv:1312.1655}},
  \href {https://doi.org/10.1016/j.jsc.2014.09.025}
  {\path{doi:10.1016/j.jsc.2014.09.025}}.

\bibitem[BH13]{Brandao2013}
Fernando~G.S.L. Brandao and Aram~W. Harrow.
\newblock Quantum de {F}inetti theorems under local measurements with
  applications.
\newblock In {\em Proceedings of the Forty-Fifth Annual ACM Symposium on Theory
  of Computing}, STOC '13, page 861–870, New York, NY, USA, 2013. Association
  for Computing Machinery.
\newblock \href {http://arxiv.org/abs/1210.6367} {\path{arXiv:1210.6367}},
  \href {https://doi.org/10.1145/2488608.2488718}
  {\path{doi:10.1145/2488608.2488718}}.

\bibitem[BH16]{Brandao2016}
Fernando~G.S.L. Brandao and Aram~W Harrow.
\newblock Product-state approximations to quantum states.
\newblock {\em Communications in Mathematical Physics}, 342(1):47--80, 2016.
\newblock \href {http://arxiv.org/abs/1310.0017} {\path{arXiv:1310.0017}},
  \href {https://doi.org/10.1007/s00220-016-2575-1}
  {\path{doi:10.1007/s00220-016-2575-1}}.

\bibitem[BMR10]{Bravyi2010}
Sergey Bravyi, Cristopher Moore, and Alexander Russell.
\newblock Bounds on the quantum satisfiability threshold.
\newblock In Andrew~Chi{-}Chih Yao, editor, {\em Innovations in Computer
  Science 2010}, pages 482--489. Tsinghua University Press, 2010.
\newblock \href {http://arxiv.org/abs/0907.1297} {\path{arXiv:0907.1297}}.

\bibitem[Bra11]{Bravyi2011}
Sergey Bravyi.
\newblock Efficient algorithm for a quantum analogue of 2-{SAT}.
\newblock In K.~Mahdavi, D.~Koslover, and L.L. Brown, editors, {\em Cross
  Disciplinary Advances in Quantum Computing}, volume 536 of {\em Contemporary
  Mathematics}. American Mathematical Society, 2011.
\newblock \href {http://arxiv.org/abs/quant-ph/0602108}
  {\path{arXiv:quant-ph/0602108}}.

\bibitem[CCD{\etalchar{+}}11]{Chen2011}
Jianxin Chen, Xie Chen, Runyao Duan, Zhengfeng Ji, and Bei Zeng.
\newblock No-go theorem for one-way quantum computing on naturally occurring
  two-level systems.
\newblock {\em Physical Review A}, 83:050301, May 2011.
\newblock \href {http://arxiv.org/abs/1004.3787} {\path{arXiv:1004.3787}},
  \href {https://doi.org/10.1103/PhysRevA.83.050301}
  {\path{doi:10.1103/PhysRevA.83.050301}}.

\bibitem[CW04]{Christandl2004}
Matthias Christandl and Andreas Winter.
\newblock Squashed entanglement: an additive entanglement measure.
\newblock {\em Journal of {M}athematical {P}hysics}, 45(3):829--840, 2004.
\newblock \href {http://arxiv.org/abs/quant-ph/0308088}
  {\path{arXiv:quant-ph/0308088}}, \href {https://doi.org/10.1063/1.1643788}
  {\path{doi:10.1063/1.1643788}}.

\bibitem[dBG16]{Beaudrap2016}
Niel de~Beaudrap and Sevag Gharibian.
\newblock {A Linear Time Algorithm for Quantum 2-SAT}.
\newblock In Ran Raz, editor, {\em 31st Conference on Computational Complexity
  (CCC 2016)}, volume~50 of {\em Leibniz International Proceedings in
  Informatics (LIPIcs)}, pages 27:1--27:21, Dagstuhl, Germany, 2016. Schloss
  Dagstuhl--Leibniz-Zentrum fuer Informatik.
\newblock \href {http://arxiv.org/abs/1508.07338} {\path{arXiv:1508.07338}},
  \href {https://doi.org/10.4230/LIPIcs.CCC.2016.27}
  {\path{doi:10.4230/LIPIcs.CCC.2016.27}}.

\bibitem[Fit18]{Fitch2018}
Matthew Fitch.
\newblock {K}ruskal-{K}atona type problem.
\newblock 2018.
\newblock \href {http://arxiv.org/abs/1805.00340} {\path{arXiv:1805.00340}}.

\bibitem[Har13]{Hartshorne2013}
R.~Hartshorne.
\newblock {\em Algebraic Geometry}.
\newblock Graduate Texts in Mathematics. Springer New York, 2013.

\bibitem[HLL{\etalchar{+}}13]{Hsu2013}
B.~Hsu, C.~R. Laumann, A.~M. L\"auchli, R.~Moessner, and S.~L. Sondhi.
\newblock Approximating random quantum optimization problems.
\newblock {\em Physical Review A}, 87:062334, 2013.
\newblock \href {http://arxiv.org/abs/1304.2837} {\path{arXiv:1304.2837}},
  \href {https://doi.org/10.1103/PhysRevA.87.062334}
  {\path{doi:10.1103/PhysRevA.87.062334}}.

\bibitem[Kee08]{Keevash2008}
Peter Keevash.
\newblock Shadows and intersections: {S}tability and new proofs.
\newblock {\em Advances in {M}athematics}, 218(5):1685--1703, 2008.
\newblock \href {http://arxiv.org/abs/0806.2023} {\path{arXiv:0806.2023}},
  \href {https://doi.org/https://doi.org/10.1016/j.aim.2008.03.023}
  {\path{doi:https://doi.org/10.1016/j.aim.2008.03.023}}.

\bibitem[KW04]{Koashi2004}
Masato Koashi and Andreas Winter.
\newblock Monogamy of quantum entanglement and other correlations.
\newblock {\em Physical Review A}, 69(2):022309, 2004.
\newblock \href {http://arxiv.org/abs/quant-ph/0310037}
  {\path{arXiv:quant-ph/0310037}}, \href
  {https://doi.org/10.1103/PhysRevA.69.022309}
  {\path{doi:10.1103/PhysRevA.69.022309}}.

\bibitem[Laz83]{lazard1983grobner}
Daniel Lazard.
\newblock {Gr{\"o}bner bases, Gaussian elimination and resolution of systems of
  algebraic equations}.
\newblock In {\em European Conference on Computer Algebra}, pages 146--156.
  Springer, 1983.

\bibitem[LLM{\etalchar{+}}10]{Laumann2010}
Christopher~R Laumann, AM~L{\"a}uchli, R~Moessner, A~Scardicchio, and
  Shivaji~Lal Sondhi.
\newblock Product, generic, and random generic quantum satisfiability.
\newblock {\em Physical Review A}, 81(6):062345, 2010.
\newblock \href {http://arxiv.org/abs/0910.2058} {\path{arXiv:0910.2058}},
  \href {https://doi.org/10.1103/PhysRevA.81.062345}
  {\path{doi:10.1103/PhysRevA.81.062345}}.

\bibitem[LMSS10]{Laumann2010a}
C.~R. Laumann, R.~Moessner, A.~Scardicchio, and S.~L. Sondhi.
\newblock Random quantum satisfiability.
\newblock {\em Quantum Information and Computation}, 10(1):1–15, Jan 2010.

\bibitem[Par04]{Parthasarathy2004}
K.R. Parthasarathy.
\newblock On the maximal dimension of a completely entangled subspace for
  finite level quantum systems.
\newblock {\em Proceedings Mathematical Sciences}, (114):365--374, 2004.
\newblock \href {http://arxiv.org/abs/quant-ph/0405077}
  {\path{arXiv:quant-ph/0405077}}, \href {https://doi.org/10.1007/BF02829441}
  {\path{doi:10.1007/BF02829441}}.

\bibitem[VZGG13]{von2013modern}
Joachim Von Zur~Gathen and J{\"u}rgen Gerhard.
\newblock {\em {Modern Computer Algebra}}.
\newblock Cambridge University Press, Cambridge, 2013.

\end{thebibliography}

\appendix
\section{Proofs of combinatorial lemmas}\label{app:combinatorics}

Throughout this appendix we consider the binomial coefficient ${x \choose k}$ to be the polynomial $x(x-1)\cdots(x-k+1)/k!$ for any nonnegative integer $k$. It is defined for every real number $x$, and is positive and increasing for $x \geq k-1$.
\begin{lemma}[Scaling binomial coefficients]\label{lem:binomialfraction}
Let $\theta, x \in \mathbb{R}_{\geq 0}$ and $c \in \mathbb{N}$, where $0 \leq \theta \leq 1$ and $\theta^{1/c}x \geq c$. Then, for some $\alpha \in [0,1)$, we have:
\begin{align} 
\theta {x \choose c}  = {\theta^{1/c} x  \choose c} + \alpha{\theta^{1/c} x  \choose c-1} < {\theta^{1/c} x + 1  \choose c} .
\end{align}
\end{lemma}
\begin{proof}
Setting $m = x-l$ for $l < x$, we have:
\[
\frac{\theta {x \choose c} - {m \choose c}}{{m \choose c-1}} = \frac{\theta x \cdots(x-c+1) }{cm\cdots(m-c+2)} - \frac{m-c+1}{c}.
\]
Now, if we set $m= \theta^{1/c}x$, we obtain:
\beas
\alpha := \frac{\theta {x \choose c} - {\theta^{1/c}x \choose c}}{{\theta^{1/c}x \choose c-1}} 
&=& \frac{1}{c} \left(\frac{\theta x\cdots(x-c+1)}{\theta^{1/c}x\cdots(\theta^{1/c}x-c+2)} - (\theta^{1/c}x-c+1) \right)
\\
&=& \frac{1}{c} \left(\frac{\theta^{1/c}x\cdots(x-c+1)}{x(x-\theta^{-1/c})\dots(x-\theta^{-1/c}(c-2))} - (\theta^{1/c}x-c+1) \right)
\\
&\leq& \frac{1}{c}(\theta^{1/c}(x-c+1)-(\theta^{1/c}x-c+1))
\\
&=& \frac{(c-1)(1-\theta^{1/c})}{c} <1.
\eeas
We also note that $\alpha \geq 0$, since $\binom{\theta^{1/c}x}{c-1}$ is positive and
\begin{align*}
\theta\binom{x}{c}-\binom{\theta^{1/c}x}{c} &= 
\frac{\theta x \cdots (x-c+1) - \theta^{1/c} x \cdots (\theta^{1/c}x-c+1)}{c!} 
\\
&= \frac{\theta^{1/c}x \cdots (\theta^{1/c}x - \theta^{1/c}(c-1)) - \theta^{1/c}x \cdots (\theta^{1/c}x - (c-1))}{c!}
\\
&\geq 0.
\end{align*}
Finally, the inequality in the statement of the lemma follows from the standard binomial recurrence relation ${x \choose c} + {x \choose c-1} = {x+1 \choose c}$.
\end{proof}

\begin{lemma*}[Restatement of Lemma~\ref{lem:finalvertex}]
    Let $G$ be a $k$-uniform hypergraph on $l$ vertices with edge density $\theta \in [0,1]$, where $\theta^{1/k}l \geq k$. If we pick a vertex at random, the probability that it lies within an edge of $G$ is at least $\theta^{1/k}$.
\end{lemma*}
\begin{proof}
    Suppose that the edge density is such that the number of edges is precisely $\binom{l'}{k}$ for some $l' < l.$ In this case the configuration of edges that minimises the probability in the statement of the lemma is clearly the complete $k$-uniform hypergraph on a subset of $l'$ vertices. If the number of edges is greater than $\binom{l'}{k}$ but less than $\binom{l'+1}{k}$, we can construct an incomplete hypergraph on $(l'+1)$ vertices, which will minimise the probability in the statement of the lemma.

    Using Lemma~\ref{lem:binomialfraction} we have 
    $$
    \binom{\theta^{1/k}l}{k} \leq \theta \binom{l}{k} \leq \binom{\theta^{1/k}l+1}{k}
    $$
    so in the minimal configuration the least vertices contained within an edge is $\theta^{1/k}l$. Therefore the probability of picking a vertex within an edge of $G$ is at least
    $$
    \frac{\theta^{1/k}l}{l} = \theta^{1/k}.
    $$
    This completes the proof.
\end{proof}

\noindent

\begin{lemma*}[Restatement of Lemma~\ref{lem:partialedgedensity}]
Let $G$ be an $l$-uniform hypergraph with edge density $\theta$, and let $k<l$ and $\omega \leq \theta$. The edge density $\theta|_{k,\omega}$ of any partial $(k,\omega)$-shadow of $G$ has the lower bound 
$$
\theta|_{k,\omega} \geq \omega \theta.
$$
\end{lemma*}

\begin{proof}
We have an $l$-uniform hypergraph $G$ on $n$ vertices, of edge density $\theta$. 
We want to find a lower bound on the edge density $\theta'$ of a $k$-uniform hypergraph $G'$ which is an $\omega$-shadow of $G$; that is, every $l$-edge of $G$ contains at least $\omega\binom{l}{k}$ $k$-edges of $G'$.

We will now relax this problem as follows. We have a total weight of $\theta'\binom{n}{k}$, which we can distribute over the $\binom{n}{k}$ possible $k$-edges of $G'$. Our goal is to distribute this weight so that, for each $l$-edge of $G$, the sum of the weights of all the $k$-edges contained in that $l$-edge is at least $\omega \binom{l}{k}$. It is clear that the original problem is a special case of this problem where the weight on each $k$-edge is constrained to be zero or one. If we can find a lower bound on $\theta'$ for this relaxation, that will therefore imply a lower bound on $\theta'$ for the original problem. 

This relaxation is a linear programming problem, defined as follows, where $\{x_v\}_{v \in \sigma_{k}([n])}$ are the weights on the $k$-edges and $\mathcal{E}$ is the set of edges of $G$:
\begin{itemize}
\item Minimise 
$$
\sum_{v \in \sigma_k([n])} x_v ,
$$
subject to the constraints
\[
\left\{\sum_{v \in \sigma_k(E)} x_v \geq \omega \binom{l}{k}, \quad x_{v} \geq 0, \quad \forall~ v \in \sigma_k([n]) \right\}_{E \in \mathcal{E}}.
\]
\end{itemize}
If we can solve this problem, we will obtain the smallest possible value for $\theta'$ by dividing the solution by $\binom{n}{k}$. Since this is a linear program, we can take the dual. In the dual the weights $\{y_E\}_{E \in \mathcal{E}}$ are assigned to the edges of $G$:
\begin{itemize}
\item Maximise
$$
\omega \binom{l}{k} \sum_{E \in \mathcal{E}} y_E,
$$
subject to the constraints
\[
\left\{\sum_{E \ni v} y_E \leq 1, \quad y_{E} \geq 0,\quad \forall~ E \in \mathcal{E} \right\}_{v \in \sigma_{k}([n]) }.
\]
\end{itemize}
The optimal solution for the dual is also the optimal solution for the primal; therefore any solution for the dual will give a lower bound for the primal. 

Now observe that the solution that uniformly assigns each edge the weight $y_E:= 1/\binom{n-k}{l-k}$ is in the feasible region for the dual, since there are at most $\binom{n-k}{l-k}$ $l$-edges containing each $k$-edge. This solution gives the lower bound 
$$
\theta' \geq \omega \theta \frac{\binom{l}{k}\binom{n}{l}}{\binom{n-k}{l-k}\binom{n}{k}} = \omega \theta,
$$
which yields the lemma.
\end{proof}

\begin{lemma*}[Restatement of Lemma~\ref{lem:oddedges}]
Let $n \in \mathbb{N}$ be odd. We construct an $(n-1)/2$-uniform hypergraph $G$ with vertex set $[n]$ as follows. For every subset $X \in \sigma_{(n-1)/2}([n])$, we do precisely one of the following:
\begin{enumerate}
    \item\label{num:option11} Add $X$ as an edge to $G$.
    \item\label{num:option12} Add all the subsets in $\sigma_{(n-1)/2}(\overline{X})$ as edges to $G$.
\end{enumerate}
The edge density $\theta$ of a hypergraph $G$ constructed in this way satisfies $\theta \geq 1/2$.
\end{lemma*}
\begin{proof}
We call a choice of~\ref{num:option11} or~\ref{num:option12} for each $X \in \sigma_{(n-1)/2}([n])$ a \emph{configuration}. For a given configuration, let $S_1 \subset \sigma_{(n-1)/2}([n])$ be the subsets for which option~\ref{num:option1} is chosen, and let $S_2 \subset \sigma_{(n-1)/2}([n])$ be the subsets for which option~\ref{num:option2} is chosen. We observe that any configuration can always be transformed into a configuration with the property that
\begin{quote}
for every $X \in S_2$, $\sigma_{(n-1)/2}(\overline{X}) \subset S_1$
\end{quote} 
without increasing the number of edges in $G$.
Indeed, suppose that this property does not hold, and there is some $Y \in S_2$ in $\sigma_{(n-1)/2}(\overline{X})$. But then one can change the configuration by moving $Y$ to $S_1$ without increasing the number of edges in $G$.

Suppose that we have a configuration obeying this property, then.  It is clear that the number of edges in $G$ is precisely $|S_1|$, and can be minimised by maximising $|S_2|$. Let $C$ be the $(n-1)/2$-uniform hypergraph on those vertices of $G$ contained in an edge of $S_1$, whose edges are precisely the edges in $S_1$. For every complete uniform subgraph of $C$ on $(n+1)/2$ vertices, we can add an edge to $S_2$, namely the complement of those vertices. By~\cite[Thm. 1]{Keevash2008}, the configuration of edges in $C$ maximising the number of induced complete subgraphs is a complete hypergraph.

We will now find out how many vertices are in $C$. Suppose that $|C|$ is less than $n-1$, and consider an subset in $\sigma_{(n-1)/2}([n])$ that contains some, but not all of the vertices in $\overline{C}$. It cannot be in $S_1$ or in $S_2$, since in that case $C$ would not contain all the vertices lying in an edge of $S_1$. We deduce that $|C|$ must be at least $n-1$. 

The following configuration therefore maximises $|S_2|$: let $|C|=n-1$ and make $C$ the complete $(n-1)/2$-uniform hypergraph on those vertices. The set $S_2$ contains all $(n-1)/2$-edges containing the single vertex in $\overline{C}$. The minimal number of edges in $G$ is therefore ${n-1 \choose (n-1)/2} = \frac{n+1}{2n}{n \choose (n-1)/2}$, and the result follows.
\end{proof}

\begin{lemma*}[Restatement of Lemma~\ref{lem:randomvertex}]
Let $G$ be a $k$-uniform hypergraph on $n$ vertices with edge density $\theta$, where $k = \beta n$ for some $\beta \in [0,1]$. Let $N$ be the degree of a vertex picked uniformly at random, and let $\tau:= N/{n-1 \choose k-1}$. Then, for any $\alpha \in(0,1)$:
\begin{align*}
\Pr\left[|\theta - \tau| \geq 
\Delta(\alpha,\beta,n)\right] \leq \alpha,
\end{align*}
where
\begin{align*}
\Delta(\alpha,\beta,n) =  \frac{5\exp\left(\frac{1}{12\alpha(1-\alpha)n}+\frac{1}{12\beta(1-\beta)n}+4\right)}{\alpha^2 \sqrt{2\pi \alpha \beta (1-\alpha) (1-\beta) n}\left(1-\frac{2}{\beta n}\right)^{5/2}\left(1-\frac{2}{(1-\alpha)(1-\beta)n}\right)^{1/2}}.
\end{align*}
\end{lemma*}

\begin{proof}
We prove the bound where $\tau < \theta$; the case $\tau > \theta$ follows by considering the complement of the hypergraph. For any hypergraph $G$ with edge density $\theta$, let $C \subset [n]$ be the subset of vertices whose degree is less than or equal to $\widetilde{\tau} {n-1 \choose k-1}$, where $\widetilde{\tau}< \theta$; clearly $\Pr\left(\widetilde{N} \leq \widetilde{\tau} {n-1 \choose k-1}\right) = |C|/n$, where $\widetilde{N}$ is the degree of a vertex of $G$ picked uniformly at random. Let $S := \sum_{v \in C} d_v$, where $d_v$ is the degree of the vertex $v \in C$; we must have $S \leq |C| \widetilde{\tau} {n-1 \choose k-1}$. We will obtain a lower bound for $\widetilde{\tau}$ in terms of $\theta$ and $|C|$ by defining a configuration of edges which minimises $S$. Since we are interested in how the lower bound changes as $n$ increases, we will consider $|C|$ to be some function of $n$; but $\theta$ will remain constant. The configuration is defined as follows: we first place all the edges contained entirely in $([n]-C)$, which does not increase $S$. Then we place all the edges which only contain one vertex in $C$, at the cost $S \mapsto S+1$. Then we can place all the edges which only contain two vertices of $C$, at the cost $S \mapsto S+2$, etc. We will stop when we have placed all $\theta {n \choose k}$ edges. Let $r+1$ be the number of vertices of $C$ in the final edge we place. The following inequality encodes the fact that we have placed $\theta {n \choose k}$ edges:
\begin{align}
\nonumber
&\sum_{i=0}^r {n -|C| \choose k-i} {|C| \choose i} \leq \theta {n \choose k} \leq \sum_{i=0}^{r+1} {n -|C| \choose k-i} {|C| \choose i} \\
\label{eq:thetaineq}
\Rightarrow ~~& P(X_{n,|C|,k} \leq r) \leq \theta \leq P(X_{n,|C|,k} \leq r+1).
\end{align}
Here $X_{n,|C|,k}$ is a random variable distributed according to the classical hypergeometric distribution with parameters $n,|C|,k$. The following inequality encodes the fact that $S \leq |C| \widetilde{\tau} {n-1 \choose k-1}$:
\begin{align}\nonumber
&\sum_{i=0}^{r} \frac{i}{|C|} {n -|C| \choose k-i} {|C| \choose i} \leq \widetilde{\tau} {n -1\choose k-1} \leq \sum_{i=0}^{r+1} \frac{i}{|C|} {n -|C| \choose k-i} {|C| \choose i}\\
\label{eq:tauineq}
\Rightarrow ~~& P(X_{n-1,|C|-1,k-1} \leq r-1) \leq \widetilde{\tau} \leq P(X_{n-1,|C|-1,k-1} \leq r).
\end{align}
Here $X_{n-1,|C|-1,k-1}$ is a random variable distributed according to the classical hypergeometric distribution with parameters $n-1,|C|-1,k-1$. It follows that:
\begin{align}\nonumber
P(X_{n,|C|,k} \leq r) &- P(X_{n-1,|C|-1,k-1} \leq r) 
\\\nonumber &\leq \theta - \widetilde{\tau} 
\\\label{eq:thetatauineq}&\leq P(X_{n,|C|,k} \leq r+1) - P(X_{n-1,|C|-1,k-1} \leq r-1).
\end{align}
For some intuition, we recall the standard model for the classical hypergeometric distribution with parameters $n, |C|,k$: there are $n$ balls, $|C|$ of which are white and $(n-|C|)$ of which are black, and we pick $k$ of them at random without replacement; the hypergeometric distribution counts the number of white balls picked. 
We observe that
\begin{align} 
P( X_{a,b,c} \leq r+1) = \frac{a-c}{a} P( X_{a-1,b-1,c} \leq r+1) + \frac{c}{a} P( X_{a-1,b-1,c-1} \leq r). \label{eq:conditionischosen1}
\end{align}
This equation is obtained by conditioning on whether a specific white ball is picked. We further observe:
\begin{align}\label{eq:conditionk}
P(X_{a,b,c} \leq r) &= P(X_{a,b,c-1} \leq r-1) + \frac{a-b+r-(c-1)}{a-(c-1)} P(X_{a,b,c-1} = r).
\end{align}
This equation is obtained by conditioning on how many of the first $c-1$ balls picked are white. Recall $k= \beta n$; let $|C| = \alpha n$. Then, substituting~\eqref{eq:conditionischosen1} and~\eqref{eq:conditionk} into~\eqref{eq:thetatauineq}, we obtain:
\begin{align*}
\theta -  \widetilde{\tau} &\leq  P(X_{n-1,|C|-1,k-1} = r) +  \frac{n-|C|+r+2-k}{n}P(X_{n-1,|C|-1,k-1} =r+1)  \\
& \leq P(X_{n-1,|C|-1,k-1} = r) + (4- \alpha - \beta) P(X_{n-1,|C|-1,k-1} =r+1).
\end{align*}
For the inequality we used that $r < n$. The mode of $X_{n-1,|C|-1,k-1}$ is $M:= \lfloor \frac{k|C|}{n+1}  \rfloor$, which implies that $M=\alpha\beta n - x$, for some $x \in [0,n/(n+1)+1]$. Then, using Stirling's approximation
$$
\sqrt{2\pi n}\left(\frac{n}{e}\right)^n <  n! < \sqrt{2\pi n}\left(\frac{n}{e}\right)^n e^{1/(12n)},
$$
after some tedious calculation we obtain
\[
\theta -  \widetilde{\tau} \leq 5\frac{{n-|C| \choose k- M}{|C| \choose M}}{{n \choose k}} \leq  \frac{5\exp\left(\frac{1}{12\alpha(1-\alpha)n}+\frac{1}{12\beta(1-\beta)n}+4\right)}{\sqrt{2\pi \beta (1-\alpha) (1-\beta) n}\left(\alpha\left(1-\frac{2}{\beta n}\right)\right)^{5/2}\left(1-\frac{2}{(1-\alpha)(1-\beta)n}\right)^{1/2}}.
\]
Since
$$
\Pr\left(\widetilde{N} \leq \widetilde{\tau}{n-1 \choose k-1}\right) = \frac{|C|}{n} = \alpha,
$$
we thus obtain the bound in the statement of the lemma.
\end{proof}

\end{document}